\documentclass[final,5p,times]{elsarticle}

\usepackage{amsmath}
\usepackage{amssymb}
\usepackage{mathtools}
\usepackage{amsthm}
\usepackage{aliascnt}
\usepackage{xcolor}
\usepackage{tikz}
\usetikzlibrary{matrix,decorations.pathreplacing}
\usepackage{quantikz}          
\usepackage{algorithm}
\usepackage{algpseudocode}
\usepackage{adjustbox}
\usepackage{float}
\usepackage{subcaption}
\usepackage{microtype}         
\usepackage{hyperref}
\usepackage{orcidlink}
\usepackage{cleveref}          
\usepackage{comment}
\theoremstyle{plain}
\newtheorem{theorem}{Theorem}
\newaliascnt{corollary}{theorem}
\newtheorem{corollary}[corollary]{Corollary}
\aliascntresetthe{corollary}
\theoremstyle{definition}
\newaliascnt{definition}{theorem}
\newtheorem{definition}[definition]{Definition}
\aliascntresetthe{definition}

\definecolor{matrixblue}{RGB}{0, 90, 180}
\definecolor{matrixfill}{RGB}{225, 235, 245}
\definecolor{matrixred}{RGB}{200, 40, 40}
\definecolor{matrixfillred}{RGB}{255, 240, 240}
\definecolor{matrixgreen}{RGB}{40, 150, 40}
\definecolor{matrixfillgreen}{RGB}{240, 255, 240}

\providecommand{\ket}[1]{\lvert #1 \rangle}
\providecommand{\bra}[1]{\langle #1 \rvert}
\providecommand{\braket}[2]{\langle #1 \mid #2 \rangle}
\newcommand{\norm}[1]{\left\| #1 \right\|}

\newcommand{\SP}{\mathcal{SP}}        
\newcommand{\MSP}{\mathcal{MSP}}      
\newcommand{\OP}{\mathcal{OP}}        
\newcommand{\Id}{\mathbb{I}}          

\journal{Future Generation Computer Systems}

\begin{document}

\begin{frontmatter}

\title{Two-Tower Quantum Matrix Chain Multiplication: Trading Qubits for Depth}

\author[csunipi]{Giacomo Antonioli\,\orcidlink{0009-0000-6687-0357}}
\ead{giacomo.antonioli@phd.unipi.it}

\author[csunipi]{Anna Bernasconi\,\orcidlink{0000-0003-0263-5221}}
\ead{anna.bernasconi@unipi.it}

\author[csunipi]{Alessandro Berti\,\orcidlink{0000-0001-9144-9572}\corref{cor1}}
\ead{alessandro.berti@df.unipi.it}

\author[csunipi]{Gianna M. Del Corso\,\orcidlink{0000-0002-5651-9368}}
\ead{gianna.delcorso@unipi.it}

\author[csunipi]{Alessandro Poggiali\,\orcidlink{0000-0002-1591-7925}}
\ead{alessandro.poggiali@di.unipi.it}

\cortext[cor1]{Corresponding author}

\affiliation[csunipi]{organization={Department of Computer Science, University of Pisa},
            city={Pisa},
            country={Italy}}

\begin{abstract}
Matrix chain multiplication --- computing $\mathcal{W} = M^{(0)}\cdots M^{(K-1)}$ where $M^{(k)} \in \mathbb{R}^{P_k \times P_{k+1}}$  --- arises in scientific computing, machine learning, and graph analysis. Despite the importance of this problem, for chains of distinct matrices, the classical number of operations grows linearly with the chain length $K$ and polynomially in the matrix dimensions. We present \emph{Two-Tower Matrix Multiplication}, a quantum subroutine that encodes the product $\mathcal{W}$ of the $K$ matrices  into a quantum state in circuit depth $\mathcal{O}(\max_{k} \mathrm{polylog} (P_k P_{k+1}))$, which is independent of~$K$ within the QRAM-based state-preparation model, whereas the qubit count is $\mathcal{O}\bigl(\sum_{k} \log P_k \bigr)$; the total gate count remains linear in $K$, so the gain is in the circuit depth. The construction interleaves state-preparation operators across two layers; within each layer, all operators act on disjoint registers and execute in parallel. This subroutine can be specialized for the chain-vector case, which computes the product of $K-1$ matrices applied to a vector. We prove the correctness of the subroutine for all $K$ and provide two implementations using the Qiskit and QCLAB frameworks. The subroutine is applicable to any downstream quantum algorithm that operates on a matrix encoded in the statevector, including norm estimation, graph-matrix powers, linear system solving, and quantum machine learning kernels.
\end{abstract}


\end{frontmatter}

\section{Introduction}
\label{sec:intro}
Matrix chain multiplication (i.e., computing
$\mathcal{W} = M^{(0)} \cdots M^{(K-1)}$ for $K$ matrices) arises in different contexts, such as
iterative solvers~\cite{GolubVanLoan1996}, deep feature maps~\cite{biamonte2017quantum}, and graph spectral
analysis~\cite{janzing610235bqp,nghiem2023quantum}.
Classically, each $N\!\times\! N$ product costs $O(N^\omega)$ ($\omega < 2.372$, due to a sequence of
improvements~\cite{Bini1979,alman2021refined,duan2023faster}), so a chain of $K$ distinct matrices requires $O(K N^\omega)$ operations. Even with unbounded parallelism (tree-structured multiplication), the depth grows as $\mathcal{O}(\log K \cdot \log N)$~\cite{GolubVanLoan1996}. Quantum linear algebra subroutines~\cite{harrow2009quantum,PhysRevLett.120.050502,biamonte2017quantum} typically promise speedups provided that an efficient state-preparation subroutine is available to encode the data into quantum states.
Two encoding models are relevant: \emph{statevector encoding}, where
matrix entries appear as amplitudes, and {block-encoding}, where the matrix is the top-left block of a unitary. Block-encoding approaches
typically have a depth that grows with $K$~\cite{chakraborty_et_al:LIPIcs.ICALP.2019.33}. In a straightforward implementation, the number of ancilla qubits grows at least linearly with $K$ because each block-encoded matrix requires its own ancilla register. The overall subnormalization factor is $\prod_{k=0}^{K-1} \alpha_k$, where $\alpha_k$ is the subnormalization of the block encoding of $M^{(k)}$ and $\alpha_k\ge \|M^{(k)}\|_2$. When this factor is large, the success probability decreases, and the computation becomes more sensitive to errors. To reduce the ancilla overhead we can apply the {compression gadget}~\cite{Fang_2023}, which lowers the qubit count to $\max_k a_k + \log(K) + 1$, where $a_k$ is the number of ancillae required to block-encode $M^{(k)}$, at the expense of increased circuit depth. Assuming efficient state-preparation, each block encoding has depth $\mathcal{O}(\mathrm{polylog}(N))$, the $K$ sequential applications yield a total circuit depth of $\mathcal{O}(K\,\mathrm{polylog}(N))$; the overall subnormalization factor however, remains unchanged. Li~et~al.~\cite{li2025faster} address the special case of single repeated matrix $A$ applied to a vector $b$, i.e., computing $A^K b$, in the block-encoding model achieving
$\Theta(\sqrt{K})$ queries via Chebyshev approximation and QSVT~\cite{gilyen2019quantum}; no analogous speedup is known for chains of {distinct} matrices via this approach.
Montanaro and Shao~\cite{montanaro2024quantum} show that for $f(x)=x^K$, the quantum query complexity of approximating an entry  $\bra{i}A^K\ket{j}$ of an $s$-sparse Hermitian matrix is $\mathcal{O}(\sqrt{K})$. This yields an exponential separation from the classical lower bound $\widetilde\Omega\!\left((s/2)^{(\sqrt{k}-1)/6}\right)$, and implying optimality of QSVT-based algorithms for matrix powers. In contrast, our setting is more general: the chain product $\mathcal{W}$ consists of matrices differing in both values and dimensions, and the  QSVT framework does not directly apply. Crucially, none of the approaches above achieves circuit depth independent of $K$ for chains of distinct matrices.

In this work, we present the Two-Tower Matrix Multiplication, a quantum subroutine that is build on the two-matrix product of~\cite{bernasconiMatMul} to chains of arbitrary length $K$. Given matrices $M^{(k)} \in \mathbb{R}^{P_k \times P_{k+1}}$ for $0 \leq k < K$, the subroutine encodes the chain product as amplitudes in circuit depth $\mathcal{O}(\max_k \mathrm{polylog}(P_k P_{k+1}))$, independent of $K$, using $\mathcal{O}\bigl(\sum_{k} \log P_k \bigr)$ qubits. We emphasize that this polylogarithmic circuit depth is achieved only when each matrix $M^{(k)}$ can itself be loaded into the quantum state in polylogarithmic time: for instance, when the entries of $M^{(k)}$  are stored in a quantum-accessible data structure such as a QRAM \cite{giovannetti2008quantum}, which supports the preparation of the corresponding quantum state in $\mathcal{O}(\operatorname{polylog}(P_k P_{k+1}))$ gates. 

Our subroutine produces an output state whose amplitudes encode the entries of the chain product $\mathcal{W}$ on a specific subspace; the remaining amplitudes are spurious terms that do not contribute to the result. The signal weight $\Omega_{\mathcal{W}} = \frac{\|\mathcal{W}\|_F^2}{ \prod_{k=0}^{K-1} \|M^{(k)}\|_F^2}$, $\Omega_{\mathcal{W}} \leq 1$, measures the total squared amplitude carried by the useful components. For well-conditioned matrices $\Omega_{\mathcal{W}}$ decays with $K$, whereas matrices with peaked spectral structure retain $\Omega_{\mathcal{W}}$ close to unity; in either case, Amplitude Amplification boosts $\Omega_{\mathcal{W}}$ to any desired constant at an additional cost of $\mathcal{O}(1/\sqrt{\Omega_{\mathcal{W}}})$ repetitions. This quantity plays the same role as the subnormalization factor $\prod_k \alpha_k$ in block-encoding approaches and subnormalization affects all quantum linear algebra subroutines, not just the Two-Tower.
Two observations inspire the proposed subroutine. (1)~Encoding the entries of a matrix as amplitudes of a quantum state is equivalent to mapping it to its row-wise vectorization.
(2)~Classical matrix theory expresses the vectorization of a three-factor product $AXB$ through the Kronecker product: $\mathrm{vec}(AXB) = (B^T \otimes A)\,\mathrm{vec}(X)$~\cite{horn2012matrix}, and this identity extends recursively to longer chains. The Two-Tower subroutine draws from this decomposition: preparing the first and last matrices of the chain, $M^{(0)}$ and $M^{(K-1)}$, on separate registers realises their Kronecker product naturally within the quantum circuit, while a contraction mechanism pairs matching indices from adjacent matrices and accumulates their products through quantum superposition, yielding the inner sums that define the chain product. The final quantum state contains the vectorization of the chain product. 

Amplitude-based quantum algorithms can leverage the Two-Tower output state directly~\cite{antonioli2025outlier,bernasconi2024quantum,poggiali2023quantum,poggiali2026more}. We observe that, as with other quantum routines such as the Quantum Fourier Transform~\cite{nielsen2010quantum}, extracting all entries by tomography negates the advantage; the subroutine is beneficial when the downstream computation requires only aggregate quantities (norms, traces, inner products) extractable via Amplitude Estimation~\cite{harrow2009quantum,PhysRevLett.120.050502}.
\begin{table}[t]
  \caption{Classical vs.\ quantum cost for the chain product of $K$ distinct square $N \times N$ matrices ($\omega < 2.372$). The classical column counts arithmetic operations; the Two-Tower column counts circuit depth. The ordering cost is the cost of finding the optimal parenthesisation of the product. The qubit count refers to circuit registers only.}
  \label{tab:complexity}
  \centering
  \begin{tabular}{lll}
    \hline
    & \textbf{Classical} & \textbf{Two-Tower} \\
    \hline
    Ordering cost          & $O(K^3)$ arith.\ ops      & not needed \\
    Execution (seq.\ depth) & $O(K \, N^\omega)$      & $\mathcal{O}({\mathrm{ polylog}} (N))$ \\
    Memory / qubits        & $O(K \, N^2)$ words     & $\Theta(K \log N)$ qubits \\
    \hline
  \end{tabular}
\end{table}
\Cref{tab:complexity} summarises the classical and quantum costs for a chain of $K$ square $N \times N$ matrices; the general rectangular case follows by replacing $N^2$ with $P_k P_{k+1}$. The qubit count $\Theta(K \log N)$ grows linearly with $K$, an exponential compression over the classical memory $O(K N^2)$, but not $K$-free. In general, the Two-Tower trades circuit depth for qubit count. We stress that this is a {depth} result: the circuit still applies $K$ state-preparation operators, so the total gate count remains $\Theta(K\,\mathrm{polylog}(N))$.

\newpage
The contributions of this work are:
\begin{itemize}
\item (1)~a subroutine with circuit depth $\mathcal{O}(\max_k \mathrm{polylog}(P_k P_{k+1}))$, independent of $K$, improving over all prior methods whose depth depends on $K$;
\item (2)~a correctness proof for arbitrary $K$;
\item (3)~a chain-vector specialisation for the case $P_K = 1$;
\item (4)~open-source Qiskit~\cite{qiskit2024} and QCLAB~\cite{keip2025qclab} implementations\footnote{\url{https://github.com/Brotherhood94/two_tower_quantum_matrix_multiplication_subroutine}}.
\end{itemize}

The remainder of this paper is organized as follows. \Cref{sec:prelim} fixes notation, recalls the QRAM-based state-preparation model, and defines the matrix chain problem. \Cref{sec:algorithm} presents the Two-Tower circuit: \Cref{sec:two_matrix} reviews the two-matrix base case~\cite{bernasconiMatMul} and its rearrangement into the two-layer structure; \Cref{sec:chain_mult} generalises to arbitrary $K$, formalises the register layout and the operators $\OP(L)$, $\OP(R)$, and provides an intuition for $K=3$ walkthrough. \Cref{sec:analysis} proves correctness for odd and even $K$, derives the depth and qubit bounds, and analyses the normalisation factor and signal weight. \Cref{sec:vector} specialises the subroutine to matrix-chain-vector products. \Cref{sec:conclusions} summarises the contributions and identifies future work.

\section{Preliminaries}
\label{sec:prelim}

This section establishes notation, introduces the QRAM-based state-preparation model, reviews the two-matrix construction of~\cite{bernasconiMatMul} that the Two-Tower generalises, and defines the matrix chain multiplication problem. We recall the key concepts; for a full treatment see~\cite{nielsen2010quantum}.

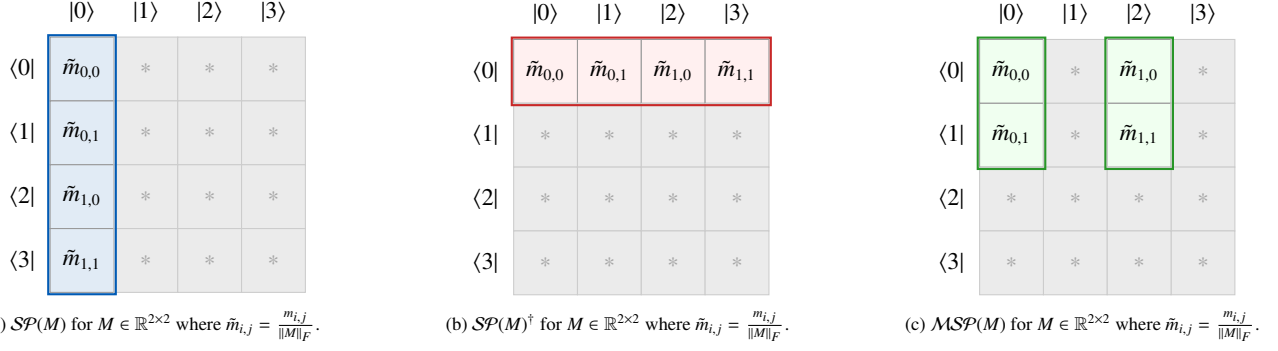
\begin{figure*}[t]
  \centering
  \begin{minipage}[t]{0.33\linewidth}
    \centering
    \begin{tikzpicture}[
      sig/.style={draw=black!40, fill=matrixfill,
                  minimum width=0.844cm, minimum height=0.844cm,
                  inner sep=0pt, font=\small, anchor=center},
      irr/.style={draw=black!20, fill=black!8,
                  minimum width=0.844cm, minimum height=0.844cm,
                  inner sep=0pt, font=\small, anchor=center,
                  text=black!35},
      klbl/.style={font=\small, text=black},
    ]
    \def\sz{0.844}
    \foreach \r in {0,...,3}{
      \node[sig] (c\r-0) at (0, -\r*\sz) {};
      \foreach \c in {1,...,3}{
        \node[irr] (c\r-\c) at (\c*\sz, -\r*\sz) {$\ast$};
      }
    }
    \node[font=\small] at (c0-0) {$\tilde{m}_{0,0}$};
    \node[font=\small] at (c1-0) {$\tilde{m}_{0,1}$};
    \node[font=\small] at (c2-0) {$\tilde{m}_{1,0}$};
    \node[font=\small] at (c3-0) {$\tilde{m}_{1,1}$};
    \node[klbl, anchor=south] at ([yshift=2pt]c0-0.north) {$\ket{0}$};
    \node[klbl, anchor=south] at ([yshift=2pt]c0-1.north) {$\ket{1}$};
    \node[klbl, anchor=south] at ([yshift=2pt]c0-2.north) {$\ket{2}$};
    \node[klbl, anchor=south] at ([yshift=2pt]c0-3.north) {$\ket{3}$};
    \node[klbl, anchor=east] at ([xshift=-2pt]c0-0.west) {$\bra{0}$};
    \node[klbl, anchor=east] at ([xshift=-2pt]c1-0.west) {$\bra{1}$};
    \node[klbl, anchor=east] at ([xshift=-2pt]c2-0.west) {$\bra{2}$};
    \node[klbl, anchor=east] at ([xshift=-2pt]c3-0.west) {$\bra{3}$};
    \draw[matrixblue, thick]
      ([xshift=-0.5pt, yshift=0.5pt]c0-0.north west)
      rectangle
      ([xshift=0.5pt, yshift=-0.5pt]c3-0.south east);
    \end{tikzpicture}
    \subcaption{$\SP(M)$ for $M\in\mathbb{R}^{2\times 2}$ where $\tilde{m}_{i,j} = \frac{m_{i,j}}{\|M\|_F}$.}
    \label{fig:viz-sp}
  \end{minipage}
  \hfill
  \begin{minipage}[t]{0.33\linewidth}
    \centering
    \begin{tikzpicture}[
      sig/.style={draw=black!40, fill=matrixfillred,
                  minimum width=0.844cm, minimum height=0.844cm,
                  inner sep=0pt, font=\small, anchor=center},
      irr/.style={draw=black!20, fill=black!8,
                  minimum width=0.844cm, minimum height=0.844cm,
                  inner sep=0pt, font=\small, anchor=center,
                  text=black!35},
      klbl/.style={font=\small, text=black},
    ]
    \def\sz{0.844}
    \foreach \c in {0,...,3}{
      \node[sig] (d0-\c) at (\c*\sz, 0) {};
      \foreach \r in {1,...,3}{
        \node[irr] (d\r-\c) at (\c*\sz, -\r*\sz) {$\ast$};
      }
    }
    \node[font=\small] at (d0-0) {$\tilde{m}_{0,0}$};
    \node[font=\small] at (d0-1) {$\tilde{m}_{0,1}$};
    \node[font=\small] at (d0-2) {$\tilde{m}_{1,0}$};
    \node[font=\small] at (d0-3) {$\tilde{m}_{1,1}$};
    \node[klbl, anchor=south] at ([yshift=2pt]d0-0.north) {$\ket{0}$};
    \node[klbl, anchor=south] at ([yshift=2pt]d0-1.north) {$\ket{1}$};
    \node[klbl, anchor=south] at ([yshift=2pt]d0-2.north) {$\ket{2}$};
    \node[klbl, anchor=south] at ([yshift=2pt]d0-3.north) {$\ket{3}$};
    \node[klbl, anchor=east] at ([xshift=-2pt]d0-0.west) {$\bra{0}$};
    \node[klbl, anchor=east] at ([xshift=-2pt]d1-0.west) {$\bra{1}$};
    \node[klbl, anchor=east] at ([xshift=-2pt]d2-0.west) {$\bra{2}$};
    \node[klbl, anchor=east] at ([xshift=-2pt]d3-0.west) {$\bra{3}$};
    \draw[matrixred, thick]
      ([xshift=-0.5pt, yshift=0.5pt]d0-0.north west)
      rectangle
      ([xshift=0.5pt, yshift=-0.5pt]d0-3.south east);
    \end{tikzpicture}
    \subcaption{$\SP(M)^\dagger$ for $M\in\mathbb{R}^{2\times 2}$ where $\tilde{m}_{i,j} = \frac{m_{i,j}}{\|M\|_F}$.}
  \end{minipage}
  \hfill
  \begin{minipage}[t]{0.33\linewidth}
    \centering
    \begin{tikzpicture}[
      sig/.style={draw=black!40, fill=matrixfillgreen,
                  minimum width=0.844cm, minimum height=0.844cm,
                  inner sep=0pt, font=\small, anchor=center},
      irr/.style={draw=black!20, fill=black!8,
                  minimum width=0.844cm, minimum height=0.844cm,
                  inner sep=0pt, font=\small, anchor=center,
                  text=black!35},
      klbl/.style={font=\small, text=black},
    ]
    \def\sz{0.844}
    \node[sig] (d0-0) at (0, -0*\sz){};
    \node[sig] (d1-0) at (0, -1*\sz){};
    \node[sig] (d0-2) at (2*\sz, -0*\sz){};
    \node[sig] (d1-2) at (2*\sz, -1*\sz){};

    \foreach \c in {1,3}{
        \node[irr] (d0-\c) at (\c*\sz, -0*\sz) {$\ast$};
        \node[irr] (d1-\c) at (\c*\sz, -1*\sz) {$\ast$};
        \node[irr] (d2-\c) at (\c*\sz, -2*\sz) {$\ast$};
        \node[irr] (d3-\c) at (\c*\sz, -3*\sz) {$\ast$};
      }
    \node[irr] (d2-0) at (0*\sz, -2*\sz) {$\ast$};
     \node[irr] (d3-0) at (0*\sz, -3*\sz) {$\ast$};
     \node[irr] (d2-2) at (2*\sz, -2*\sz) {$\ast$};
      \node[irr] (d2-3) at (2*\sz, -3*\sz) {$\ast$};

    \node[font=\small] at (d0-0) {$\tilde{m}_{0,0}$};
    \node[font=\small] at (d1-0) {$\tilde{m}_{0,1}$};
    \node[font=\small] at (d0-2) {$\tilde{m}_{1,0}$};
    \node[font=\small] at (d1-2) {$\tilde{m}_{1,1}$};
    \node[klbl, anchor=south] at ([yshift=2pt]d0-0.north) {$\ket{0}$};
    \node[klbl, anchor=south] at ([yshift=2pt]d0-1.north) {$\ket{1}$};
    \node[klbl, anchor=south] at ([yshift=2pt]d0-2.north) {$\ket{2}$};
    \node[klbl, anchor=south] at ([yshift=2pt]d0-3.north) {$\ket{3}$};
    \node[klbl, anchor=east] at ([xshift=-2pt]d0-0.west) {$\bra{0}$};
    \node[klbl, anchor=east] at ([xshift=-2pt]d1-0.west) {$\bra{1}$};
    \node[klbl, anchor=east] at ([xshift=-2pt]d2-0.west) {$\bra{2}$};
    \node[klbl, anchor=east] at ([xshift=-2pt]d3-0.west) {$\bra{3}$};
    \draw[matrixgreen, thick]
      ([xshift=-0.5pt, yshift=0.5pt]d0-0.north west)
      rectangle
      ([xshift=0.5pt, yshift=-0.5pt]d1-0.south east);
      \draw[matrixgreen, thick]
      ([xshift=-0.5pt, yshift=0.5pt]d0-2.north west)
      rectangle
      ([xshift=0.5pt, yshift=-0.5pt]d1-2.south east);
    \end{tikzpicture}
    \subcaption{$\MSP(M)$ for $M\in\mathbb{R}^{2\times 2}$ where $\tilde{m}_{i,j} = \frac{m_{i,j}}{\|M\|_F}$.}
  \end{minipage}
  \caption{Unitary matrices corresponding to
    \textbf{(a)}~$\SP(M)$,   \textbf{(b)}~$\SP(M)^\dagger$, and \textbf{(c)}~$\MSP(M)$
    shown schematically for $2\times 2$ matrix $M$.
    In $\SP(M)$, the first column (blue) contains the vectorised,
    Frobenius-normalised entries of $M$; the remaining columns (grey, marked $\ast$) are fixed by unitarity and play no role in the algorithm.
    In $\SP(M)^\dagger$, the first row (red) contains the normalised entries of $M$; the remaining rows are irrelevant. In $\MSP(M)$, the normalised entries (green) of $M$ retain the row structure of the original matrix.
    }
  \label{fig:viz}
\end{figure*}
\subsection{Notation}
\label{sec:prelim:notation}
An {$n$-qubit quantum state} is a unit vector $\ket{\psi} \in \mathbb{C}^{2^n}$ and $\ket{i}^n$ denotes the computational-basis state encoding $i \in [0, 2^n)$, with the superscript omitted when clear from context. An {$n$-qubit unitary} (also called a {quantum gate} or {quantum operator}) is a matrix $U \in \mathbb{C}^{2^n \times 2^n}$ satisfying $UU^\dagger = \Id_{2^n}$. We fix the notation $\ket{\psi}^{\mathrm{size}}$ where the superscript denotes the number of qubits and is omitted when clear from context.

For a matrix $A \in \mathbb{R}^{P \times R}$ with entries $a_{i,j}$, we write $\norm{A}_F = \sqrt{\sum_{i,j} a_{i,j}^2}$ for its Frobenius norm, and $\norm{A_i}$ for the Euclidean norm of the $i$-th row of $A$. We use lowercase letters for qubit counts: $p = \log P$, $r = \log R$, and more generally $p_k = \log P_k$, for a matrix chain. Throughout the paper, all matrices have real entries and all dimensions are assumed to be powers of two. 
We denote the conjugate transpose of a matrix or operator with the dagger symbol ($\dagger$), and we refer to it as the adjoint of the matrix or operator.
We now formalize the matrix chain multiplication problem.

\begin{definition}[Matrix Chain Product]\label{def:chain}
  Let $K \geq 2$.
  A {matrix chain} is a sequence of matrices
  $M^{(0)}, \ldots, M^{(K-1)}$ where each
  $M^{(k)} \in \mathbb{R}^{P_k \times P_{k+1}}$ with
  $P_k = 2^{p_k}$, 
  for $k = 0, \ldots, K-1$.
  Denoting by $m^{(k)}_{i,j}$ the $(i,j)$-th entry of $M^{(k)}$,
  the {chain product} $\mathcal{W} = M^{(0)}\cdots M^{(K-1)} \in
  \mathbb{R}^{P_0 \times P_K}$ has $(p,q)$-th entry $w_{p,q}$
  given by
\begin{equation*}
    w_{p,q} =
      \sum_{r_0,\ldots,r_{K-2}}
      \underbrace{m^{(0)}_{p,\,r_0}}_{\text{outer index }p}
      \underbrace{ m^{(1)}_{r_{0},\,r_1}m^{(2)}_{r_{1},\,r_2}\cdots m^{(K-2)}_{r_{K-3},\,r_{K-2}}}_{\text{shared indices}}
      \underbrace{m^{(K-1)}_{r_{K-2},\,q}}_{\text{outer index }q},
  \end{equation*}
  where $p \in [0,P_0)$ and $q \in [0,P_{K})$ are the fixed outer
  indices, and each shared index $r_j$ ranges over $[0, P_{j+1})$.
\end{definition}

The indices $r_0, \ldots, r_{K-2}$ are {shared indices}; they are
summed over in the chain product.
The outer indices $p \in [0, P_0)$ and $q \in [0, P_K)$ label the
entries of $\mathcal{W}$.

\subsection{State preparation}
\label{sec:prelim:sp}
To build quantum algorithms for linear algebra, one must first
efficiently load classical data, i.e., matrix entries,  into quantum states.
The dominant bottleneck is the {state-preparation step}: encoding
an $N$-dimensional real vector into the amplitudes of a
$\lceil\log N\rceil$-qubit state.
Throughout this work, state preparation is assumed to operate under the
{Quantum Random Access Memory} (QRAM) model
of~\cite{berti2025efficient, kerenidis_et_al:LIPIcs.ITCS.2017.49}. More precisely, a QRAM is a memory device that holds classical data and implements a mapping of the form:
\begin{equation}
    \mathsf{QRAM:}\ket{i}^{n} \ket{b}^{t} \mapsto  \ket{i}^{n} \ket{b \oplus x_i}^{t},
\label{eq:QRAM_access_transformation}
\end{equation}
where the register $\ket{i}^{n}$ (size $n$, $i \in \mathbb{N}$) holds the address and spans $N=2^n$ distinct locations, while $\ket{b \oplus x_i}^{t}$ is the data register: $b$ is an arbitrary bit string and $x_i \in \{0,1\}^t$ is the $t$-bit representation of the content stored at address $\ket{i}^n$. What distinguishes this memory model is that it acts on an address register prepared in superposition and returns the corresponding superposition of stored values. For state preparation, we store the vector (or matrix) entries in the quantum-accessible data structure of~\cite{kerenidis_et_al:LIPIcs.ITCS.2017.49}: a binary tree whose leaves hold the signed entries and whose internal nodes hold partial squared norms; a classical preprocessing step builds this tree once, costing $\mathcal{O}(PR)$ operations for a $P \times R$ matrix.
Under this model, a quantum state encoding an $N$-dimensional real vector can be prepared in depth $\mathcal{O}(\mathrm{polylog} (N))$ using
$\Theta(\log N)$ qubits.
Alternative data-loading models, such as block-encoding~\cite{chakraborty_et_al:LIPIcs.ICALP.2019.33}, produce a different data representation.
All complexity bounds in this paper are stated within the QRAM model. With the QRAM model in place, we now define the state-preparation
operator for a matrix.

\begin{definition}[State-Preparation Operator]\label{def:sp}
  Let $M \in \mathbb{R}^{P \times R}$ with $P = 2^p$ and $R = 2^r$.
  The \emph{state-preparation operator} $\SP(M)$ is a unitary on
  $p + r$ qubits that maps
  \begin{equation}\label{eq:sp}
    \SP(M)\colon \ket{0}^{p}\ket{0}^{r}
    \;\longmapsto\;
    \frac{1}{\norm{M}_F}
    \sum_{i=0}^{P-1}\sum_{j=0}^{R-1} m_{i,j}\,\ket{i}^{p}\ket{j}^{r}.
  \end{equation}
  The operator decomposes as $\SP(M) = U(M) V(M)$, where
  $V(M)$ prepares the row-weight state
  $\frac{1}{\norm{M}_F}\sum_{i=0}^{P-1} \norm{M_i}\,\ket{i}$, and $U(M)$ loads the normalised row
  $\frac{M_i}{\norm{M_i}}$~\cite{kerenidis_et_al:LIPIcs.ITCS.2017.49}.
  Under the QRAM model, $\SP(M)$ can be realised in depth $\mathcal{O}(\mathrm{polylog}(PR))$ using
  $\Theta(\log(PR))$ qubits.
  The adjoint satisfies
  $\SP(M)^\dagger = V(M)^\dagger  U(M)^\dagger$.
\end{definition}

\begin{figure*}[t]
  \centering
  \begin{minipage}[b]{0.48\linewidth}
    \centering
    \resizebox{\linewidth}{!}{%
    \begin{quantikz}[wire types={b,b,b}, classical gap=0.07cm]
      \lstick{$\ket{0}^{n}$}
        & \qw
        & \swap{1}
        & \qw
        & \qw
        & \swap{1}
        & \qw
        & \qw \\
      \lstick{$\ket{0}^{m}$}
        & \gate[wires=2]{V(A)}
        & \targX{}
        & \gate[wires=2]{V(B^T)}
        & \gate[wires=2]{U(B^T)}
        & \targX{}
        & \gate[wires=2]{U(A)^\dagger}
        & \qw \\
      \lstick{$\ket{0}^{s}$}
        & \qw
        & \qw
        & \qw
        & \qw
        & \qw
        & \qw
        & \qw
    \end{quantikz}}
    \subcaption{Original circuit from Bernasconi~et~al.~\cite{bernasconiMatMul}.}
    \label{fig:circuit-original}
  \end{minipage}
  \hfill
  \begin{minipage}[b]{0.48\linewidth}
    \centering
    \resizebox{0.59\linewidth}{!}{%
    \begin{quantikz}[wire types={b,b,b}, classical gap=0.07cm]
      \lstick{$\ket{0}^{s}$}
        & \qw\gategroup[3,steps=1,style={dotted, rounded corners, inner sep=1pt},label style={yshift=-3.8cm}]{$\OP(L)$}
        & \gate[wires=2]{\MSP(B)}\gategroup[3,steps=1,style={dotted, rounded corners, inner sep=1pt},label style={yshift=-3.8cm}]{$\OP(R)$}
        & \qw \\
      \lstick{$\ket{0}^{n}$}
        & \gate[wires=2]{\SP(A)}
        & \qw
        & \qw \\
      \lstick{$\ket{0}^{m}$}
        & \qw
        & \qw
        & \qw
    \end{quantikz}}
    \subcaption{Equivalent Two-Tower circuit, reorganised into layers $\OP(L)$ and $\OP(R)$.}
    \label{fig:circuit-2}
  \end{minipage}
  \caption{Two equivalent circuits for computing $C = AB$, with
    $A \in \mathbb{R}^{M \times N}$, $B \in \mathbb{R}^{N \times S}$,
    $m = \log M$, $n = \log N$, $s = \log S$.
    \textbf{(a)}~Original circuit from
    Bernasconi~et~al.~\cite{bernasconiMatMul}, using four
    state-preparation operators ($V(A)$, $U(A)^\dagger$, $V(B^T)$,
    $U(B^T)$) and performing two register swaps.
    \textbf{(b)}~The same circuit reorganised into the Two-Tower
    two-layer structure $\OP(L)$ and $\OP(R)$: layer~$\OP(L)$ applies
    $\SP(A) = U(A) \, V(A)$ on registers $(n,m)$; layer~$\OP(R)$
    applies $\MSP(B)$ on registers $(s,n)$, a composite operator that
    applies the row-loading multiplexer $U(B)$ on $(s,n)$ followed by
    $V(B)^\dagger$ on $n$ to contract the shared index. This
    reorganisation eliminates the register swaps and constitutes the
    base case ($K=2$) of the Two-Tower construction.}
  \label{fig:circuit-base}
\end{figure*}
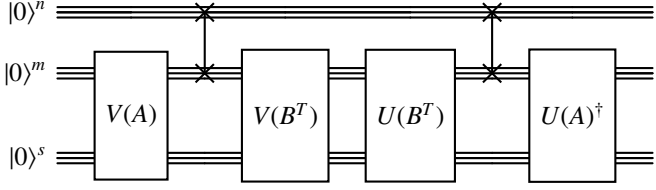
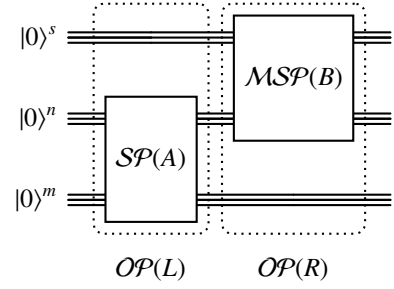

When the chain length $K$ is even, the state-preparation operator of~\Cref{def:sp} alone does not suffice to encode the last matrix. We therefore introduce the following operator.

\begin{definition}[Modular-State-Preparation Operator]\label{def:MSP}
  Let $M \in \mathbb{R}^{P \times R}$ with $P = 2^p$ and $R = 2^r$. The \emph{Modular state-preparation operator} is a unitary on $p+r$ qubits such that 
  $$
  \MSP(M)\colon \ket{i}^p\ket{0}^r \;\longmapsto\;
  \frac{1}{\|M\|_F} \sum_{k=0}^{R-1} m_{i, k} \ket{0}^p \ket{k}^r +\ket{\bot}.
  $$
  Given $\SP(M)=U(M)V(M)$ the state-preparation operator defined in Definition~\ref{def:sp}, the Modular state-preparation operator can be decomposed as $\MSP(M)=V(M)^\dagger U(M)$. The term $\ket{\bot}$ denotes the superposition of all the other computational basis states, where the $p$-qubit register is not $\ket{0}^p$.
\end{definition}

\Cref{fig:viz} illustrates the internal structure of the $\SP$, $\SP^\dagger$, and $\MSP$ unitaries for a concrete small matrix $M$. By \Cref{def:sp}, the first column of $\SP(M)$ contains the vectorised, Frobenius-normalised entries of $M$; all remaining columns are fixed by unitarity but play no role in the algorithm. Dually, the first row of $\SP(M)^\dagger$ contains the normalised entries of $M$; all remaining rows are irrelevant. Finally, in $\MSP(M)$, the rows of the matrix $M$ are distributed among the columns with index a multiple of $R$. In particular, $\tilde m_{i,j}= \bra{0}^p\bra{j}^r\MSP(M)\ket{i}^p\ket{0}^r,$ with $i=0, \ldots, P-1$ and $j=0, \ldots, R-1.$ All remaining entries of $\MSP(M)$ are irrelevant. While the standard state-preparation operator $\SP(M)$ encodes the row-wise vectorization of the matrix $M$ into its first column, the modular state-preparation operator $\MSP(M)$ preserves the matrix's two-dimensional structure, encoding each row of $M$ in a separate column.

\section{The Two-Tower Algorithm}
\label{sec:algorithm}

This section presents the Two-Tower algorithm in two stages.
\Cref{sec:two_matrix} recalls the two-matrix subroutine
of~\cite{bernasconiMatMul}, which serves as the foundation. 
\Cref{sec:chain_mult} generalizes this construction to an arbitrary
chain of $K$ matrices.

\subsection{Two-matrix multiplication}\label{sec:two_matrix}
In~\cite{bernasconiMatMul}, the authors introduce a quantum
subroutine for the product of two matrices
$A \in \mathbb{R}^{M \times N}$ and $B \in \mathbb{R}^{N \times S}$.
Their original circuit, shown in \Cref{fig:circuit-original}, uses
operators $V(A)$, $U(A)^\dagger$, $V(B^T)$, $U(B^T)$ together with
two register swaps on quantum registers $\ket{0}^n$, $\ket{0}^m$, and $\ket{0}^s$.

We observe that this circuit can be rearranged
into the equivalent form shown in \Cref{fig:circuit-2}, which serves
as the base case ($K = 2$) of the Two-Tower construction.
The key insight is that the composition
$U(A)^\dagger \, \,\mathrm{SWAP} \, U(B^T)  V(B^T) \, \, 
\mathrm{SWAP}\,  V(A)$ can be reorganised as the composition of 
$\MSP(B)$ and $\SP(A)$, eliminating the register swaps and
grouping the operators into two sequential layers acting on disjoint
register pairs. This reorganisation, while functionally identical for $K = 2$, reveals a two-layer structure --- denoted by $\OP(L)$ and $\OP(R)$ --- that generalises to arbitrary $K$. In particular, the rearranged circuit acts on $m + n + s$ qubits and produces:
\begin{multline*}
(\mathbb{I}_M \otimes\MSP(B))\,(\SP(A)\otimes \mathbb{I}_S)\,\ket{0}^{m+n+s} = \\
=\frac{1}{\norm{A}_F\norm{B}_F}
    \sum_{i=0}^{M-1}\sum_{k=0}^{S-1}
    \underbrace{\sum_{j=0}^{N-1} a_{i,j}\,b_{j,k}}_{(AB)_{i,k}}
    \ket{i}^m\ket{0}^n\ket{k}^s + \ket{\bot},
\end{multline*}
encoding the product $AB$ in the outer registers with the internal
index contracted to $\ket{0}^n$ (see \Cref{subsec:proof-even} for details).

\subsection{Chain multiplication for arbitrary \texorpdfstring{$K$}{K}}\label{sec:chain_mult}
The two-matrix construction of \Cref{sec:two_matrix} generalises to an arbitrary chain of $K$ matrices
$\mathcal{W} = M^{(0)} M^{(1)} \cdots M^{(K-1)}$.
Before describing the full algorithm, we isolate the mechanism that
makes the construction work.

\subsubsection{The contraction mechanism}
Recall from \Cref{def:sp} that $\SP(M)$ maps $$\ket{0,0} \mapsto \frac{1}{\norm{M}_F}\sum_{i, j} m_{i, j}\ket{i, j}.$$ The adjoint $\SP(M)^\dagger$ reverses this map. Moreover, when $\SP(M)^\dagger$ acts on a basis state $\ket{i,j}$, it projects onto $\ket{0,0}$ with amplitude proportional to the corresponding matrix entry $m_{i,j}$. Formally,
\begin{equation}\label{eq:sp-collapse}
    \SP(M)^{\dagger}\ket{i,j}
    \;=\;
    \frac{m_{i,j}}{\norm{M}_F}\,\ket{0,0}
    \;+\; \ket{\bot}, \quad   0 \leq i < P, \text{ and } 0 \leq j < R,
\end{equation}
where $\ket{\bot}$ is orthogonal to $\ket{0,0}$. We observe that if two consecutive matrices $M^{(k)}$ and $M^{(k+1)}$ share a summation index encoded in the same quantum register, applying $\SP(M^{(k)})$ populates that register with amplitudes proportional to the entries of $M^{(k)}$, and $\SP(M^{(k+1)})^\dagger$ contracts it back to $\ket{0}$, multiplying each amplitude by the corresponding entry of $M^{(k+1)}$.
By linearity, all such combinations occur simultaneously, and the resulting amplitudes reproduce exactly the classical matrix product. The Two-Tower circuit arranges all matrices so that every shared index is contracted through this mechanism.
\begin{figure*}[t]
  \centering

  \begin{minipage}{0.44\linewidth}
    \centering
    \newsavebox{\circOdd}%
    \sbox{\circOdd}{%
      \resizebox{\linewidth}{!}{%
        \begin{quantikz}[wire types={b,b,b,b,b,b,b,b}, classical gap=0.07cm]
                \lstick{$\ket{0}^{r_{\frac{K-1}{2}}}$} & \gate[wires=2]{\mathrm{SP}(M^{(K-1)})}\gategroup[8,steps=1,style={dotted,rounded corners,inner sep=1pt}, label style={yshift=-9.5cm}]{$\OP(L)$} & \gate[style={inner sep=2pt}]{\mathbb{I}_{R_{\frac{K-1}{2}}}}\gategroup[8,steps=1,style={dotted,rounded corners,inner sep=1pt}, label style={yshift=-9.5cm}]{$\OP(R)$} & \qw \\
                \lstick{$\ket{0}^{l_{\frac{K-1}{2}}}$} & \qw & \gate[wires=2]{\mathrm{SP}(M^{(K-2)})^{\dagger}} & \qw \\
                \lstick{$\ket{0}^{r_{\frac{K-1}{2}-1}}$} & \qw & \qw & \qw \\
                & \wave & \wave & \\
                \lstick{$\ket{0}^{r_1}$}
                & \gate[wires=2]{\mathrm{SP}(M^{(2)})} & \qw & \qw \\
                \lstick{$\ket{0}^{l_1}$} & \qw & \gate[wires=2]{\mathrm{SP}(M^{(1)})^{\dagger}} & \qw \\
                \lstick{$\ket{0}^{r_0}$} & \gate[wires=2]{\mathrm{SP}(M^{(0)})} & \qw & \qw \\
                \lstick{$\ket{0}^{l_0}$} & \qw & \gate[style={inner sep=2pt}]{\mathbb{I}_{L_0}} & \qw
        \end{quantikz}}%
    }%
    \begin{tikzpicture}
      \node[inner sep=0pt, outer sep=0pt] (circ) {\usebox{\circOdd}};
      \begin{scope}[
          shift={(circ.south west)},
          x    ={($(circ.south east)-(circ.south west)$)},
          y    ={($(circ.north west)-(circ.south west)$)}]

        \fill[white, opacity=1 ] (0.28, 0.75) rectangle (0.55, 0.65);
        \node at (0.425, 0.68){\scalebox{0.7}{$\mathrm{SP}(M^{(K-3)})$}};
        \draw[black, thick] (0.28, 0.615) -- (0.28, 0.75) -- (0.55, 0.75) -- (0.55, 0.62);

        \fill[white, opacity=1] (0.64, 0.45) rectangle (0.887, 0.55);
        \node at (0.78, 0.52){\scalebox{0.9}{$\mathrm{SP}(M^{(3)})^{\dagger}$}};
        \draw[black, thick] (0.64, 0.605) -- (0.64, 0.45) -- (0.885, 0.45) -- (0.885, 0.608);
      \end{scope}
    \end{tikzpicture}
    \subcaption{$K$ odd.}
    \label{fig:circuit-K-odd}
  \end{minipage}
  \hfill
  \begin{minipage}{0.44\linewidth}
    \centering
    \newsavebox{\circEven}%
    \sbox{\circEven}{%
      \resizebox{\linewidth}{!}{%
        \begin{quantikz}[wire types={b,b,b,b,b,b,b,b}, classical gap=0.07cm]
                \lstick{$\ket{0}^{l_{\frac{K}{2}}}$} & \gate[style={inner sep=2pt}]{\mathbb{I}_{L_{\frac{K}{2}}}}\gategroup[8,steps=1,style={dotted,rounded corners,inner sep=1pt}, label style={yshift=-9.5cm}]{$\OP(L)$} & \gate[wires=2]{\MSP(M^{(K-1)})}\gategroup[8,steps=1,style={dotted,rounded corners,inner sep=1pt}, label style={yshift=-9.5cm}]{$\OP(R)$} & \qw \\
                \lstick{$\ket{0}^{r_{\frac{K}{2}-1}}$} & \gate[wires=2]{\mathrm{SP}(M^{(K-2)})} & \qw & \qw \\
                \lstick{$\ket{0}^{l_{\frac{K}{2}-1}}$} & \qw & \qw & \qw \\
                & \wave & \wave & \\
                \lstick{$\ket{0}^{r_1}$}
                & \gate[wires=2]{\mathrm{SP}(M^{(2)})} & \qw & \qw \\
                \lstick{$\ket{0}^{l_1}$} & \qw & \gate[wires=2]{\mathrm{SP}(M^{(1)})^{\dagger}} & \qw \\
                \lstick{$\ket{0}^{r_0}$} & \gate[wires=2]{\mathrm{SP}(M^{(0)})} & \qw & \qw \\
                \lstick{$\ket{0}^{l_0}$} & \qw & \gate[style={inner sep=2pt}]{\mathbb{I}_{L_0}} & \qw
        \end{quantikz}}%
    }%
    \begin{tikzpicture}
      \node[inner sep=0pt, outer sep=0pt] (circ) {\usebox{\circEven}};
      \begin{scope}[
          shift={(circ.south west)},
          x    ={($(circ.south east)-(circ.south west)$)},
          y    ={($(circ.north west)-(circ.south west)$)}]
         
         \fill[white, opacity=1] (0.63, 0.75) rectangle (0.88, 0.66);
        \node at (0.755, 0.675){\scalebox{0.7}{$\mathrm{SP}(M^{(K-3)})^{\dagger}$}};
        \draw[black, thick] (0.63, 0.629) -- (0.63, 0.75) -- (0.88, 0.75) -- (0.88, 0.63);

        \fill[white, opacity=1] (0.627, 0.45) rectangle (0.864, 0.55);
        \node at (0.76, 0.52){\scalebox{0.9}{$\mathrm{SP}(M^{(3)})^{\dagger}$}};
        \draw[black, thick] (0.627, 0.606) -- (0.627, 0.45) -- (0.864, 0.45) -- (0.864, 0.608);
\draw[red, thick, dashed] 
    (0.2, 0.08) -- 
    (0.2, 0.875) -- 
    (0.55, 0.875) -- 
    (0.55, 0.765) -- 
    (0.96, 0.765) -- 
    (0.96, 0.08) -- 
    cycle;          
      \end{scope}
    \end{tikzpicture}
    \subcaption{$K$ even.}
    \label{fig:circuit-K-even}
  \end{minipage}
  \caption{Two-Tower circuit for a chain of $K$ matrices.
    \textbf{(a)}~$K$ odd: registers are ordered top-to-bottom as
    $(r_{(K-1)/2},\, l_{(K-1)/2},\, \ldots,\, r_0,\, l_0)$.
    $\OP(L)$: all $\SP(M^{(2i)})$ act simultaneously on
    disjoint register pairs $(r_i, l_i)$.
    $\OP(R)$: all $\SP(M^{(2j+1)})^\dagger$ act simultaneously
    on the interleaved pairs $(l_{j+1}, r_j)$.
    \textbf{(b)}~$K$ even: an additional output register $l_{\frac{K}{2}}$
    appears at the top.
    The last matrix $M^{(K-1)}$ uses the $\MSP$ state preparation
    (denoted $\MSP(M^{(K-1)})$), which acts on
   $(l_{\frac{K}{2}},\, r_{K/2-1})$ within $\OP(R)$. This case can also be expressed as the product of two matrices: the result of applying the odd-case logic to the chain of the first $K-1$ matrices (dashed box in red), multiplied by the final matrix $M^{(K-1)}$ throught the $\MSP$ operator.}
  \label{fig:circuit-K}
\end{figure*}
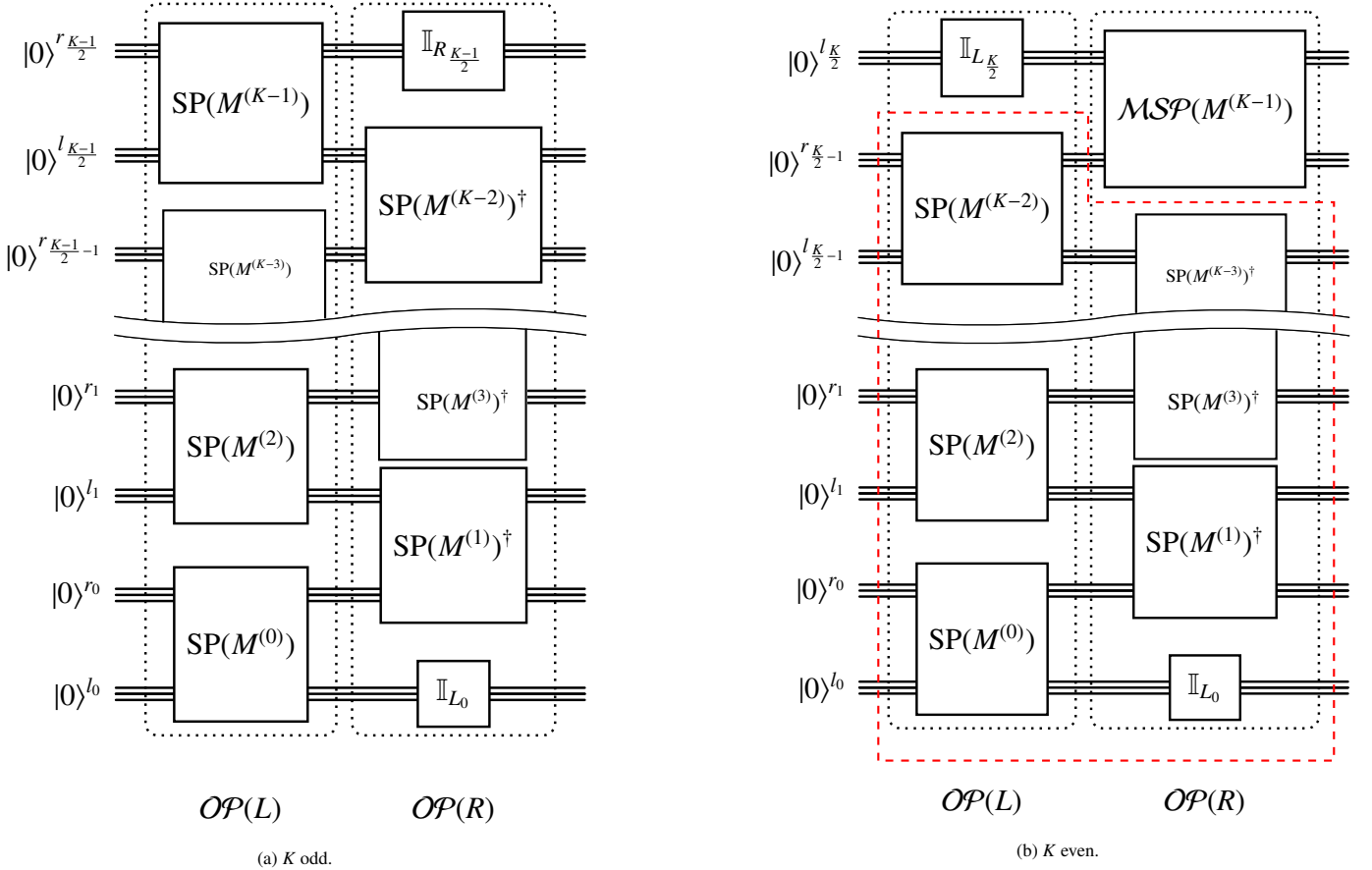

\subsubsection{Circuit structure}
The circuit consists of two parallel layers (see \Cref{fig:circuit-K}):
$\OP(L)$ applies $\SP$ to every even-indexed matrix, and
$\OP(R)$ applies $\SP^\dagger$ to every odd-indexed matrix, with one exception: when $K$ is even, the last matrix $M^{(K-1)}$ has no paired right neighbour, and $\MSP(M^{(K-1)})$ is applied in place of $\SP(M^{(K-1)})^\dagger$ so that the output column register is preserved (see~\Cref{def:MSP} and~\Cref{fig:circuit-K-even}). Since the operators within each layer act on non-overlapping registers,
they execute in parallel, and the total circuit depth is independent
of $K$. Following the intuition from the matrix-product vectorization discussed in the introduction, these two layers realise the Kronecker products of the first and last matrices of each internal subchain; \Cref{sec:algorithm:walkthrough} illustrates this in detail.
\begin{algorithm*}[t]
\caption{Two-Tower Quantum Matrix Chain Multiplication}\label{alg:skyscraper}
\begin{algorithmic}[1]
  \Require Matrix chain $\mathcal{W} = (M^{(0)},\ldots,M^{(K-1)})$
           with $M^{(k)} \in \mathbb{R}^{P_k \times P_{k+1}}$, all dimensions powers of two.
  \Ensure Quantum state encoding $\mathcal{W} = M^{(0)}\cdots M^{(K-1)}$
          up to Frobenius-norm scaling.
  \State Allocate register $q_k$ of $\log P_k$ qubits for each
         $k = 0,\ldots,K-1$, and a final register $q_K$ of
         $\log P_K$ qubits; initialise all to $\ket{0}$.
    \State \textbf{Layer 1 (parallel).} For each even $k=2i,\;  0 \leq i \leq \lfloor(K-1)/2\rfloor$:
         apply $\SP(M^{(k)})$ on registers $(q_{k+1}, q_k)$.
    \State \textbf{Layer 2 (parallel).} For each odd $k=2j+1,\;  0 \leq j \leq\lceil K/2 -2\rceil$:
         apply $\SP(M^{(k)})^\dagger$ on registers $(q_{k+1}, q_k)$.
  \If{$K$ is even}
    \State Apply $\MSP(M^{(K-1)})$ on $(q_K, q_{K-1})$.
  \EndIf
  \State \Return The state $U_\mathcal{W}\ket{0} 
        $
        $q_0$ holds the row index of $\mathcal{W}$, $q_K$ holds the column index,
         and registers $q_1,\ldots,q_{K-1}$ are in state $\ket{0}$.
\end{algorithmic}
\end{algorithm*}

We observe that each matrix requires two registers to index rows and columns. However, registers are shared between adjacent matrices, so the overall circuit requires $K+1$  registers: two for the first matrix and one additional register for each remaining matrix.
Thus, the circuit grows in qubit count while its depth remains constant --- much like two towers rising floor by floor, which gives the algorithm its name. In~\Cref{def:skyscraper-ops} we formalise the construction.

\begin{definition}[Two-Tower operators]\label{def:skyscraper-ops}
  Let $\mathcal{W} = M^{(0)}, \ldots, M^{(K-1)}$ be a matrix chain with $K \ge 2$.
Partition the chain into:
  \begin{itemize}
    \item \emph{Left matrices} (even index):
          $M^{(2i)} \in \mathbb{R}^{L_i \times R_i}$, for
          $0 \leq i \leq \lfloor(K-1)/2\rfloor$;
    \item \emph{Right matrices} (odd index):
          $M^{(2j+1)} \in \mathbb{R}^{R_j \times L_{j+1}}$, for
          $0 \leq j \leq \lfloor(K/2 - 1)\rfloor$.
  \end{itemize}
  Let us define $l_i = \log L_i$ and $r_i = \log R_i$ for the qubit counts.
  For each left matrix allocate a register pair
  $\ket{0}^{l_i}\ket{0}^{r_i}$; the initial state is
\begin{equation*}\label{eq:psi0}
    \ket{\psi_0}
    = \left\{\begin{array}{ll}
    \bigotimes_{i=0}^{(K-1)/2} \ket{0}^{l_i}\ket{0}^{r_i}, & \mbox{if $K$ is odd}\\[9pt]
    \left(\bigotimes_{i=0}^{K/2-1} \ket{0}^{l_i}\ket{0}^{r_i}\right) \ket{0}^{l_{K/2}}, & \mbox{if $K$ is even.}
    \end{array}
    \right.
\end{equation*}
The circuit uses $\sum_{i=0}^{(K-1)/2}(l_i + r_i)$ qubits if $K$ is odd and $\sum_{i=0}^{K/2-1}(l_i + r_i) + l_{K/2}$ if $K$ is even. 

Define:
  \begin{align}
    \OP(L) &= \left\{\begin{array}{ll}
    \bigotimes_{i=0}^{(K-1)/2} \SP\!\left(M^{(2i)}\right),& \mbox{if $K$ is odd}\\[9pt]
    \bigotimes_{i=0}^{K/2-1} \SP\!\left(M^{(2i)}\right) \otimes \Id_{l_\frac{K}{2}}, & \mbox{if $K$ is even;}\\[9pt]
    \end{array}
    \right.
    \label{eq:OPL}\\[4pt]
    \OP(R) &= \adjustbox{max width=0.82\linewidth}{$\displaystyle\left\{\begin{array}{ll} \underbrace{\Id_{L_0}}_{\text{left boundary}}
              \otimes
              \underbrace{\bigotimes_{j=0}^{(K-3)/2} \SP\!\left(M^{(2j+1)}\right)^\dagger}_{\text{contracts shared indexes}}
              \otimes
              \underbrace{\Id_{R_{(K-1)/2}}}_{\text{right boundary}} & \mbox{if $K$ is odd}\\
              \underbrace{\Id_{L_0}}_{\text{left boundary}}
              \otimes
              \underbrace{\bigotimes_{j=0}^{K/2-2}\SP\!\left(M^{(2j+1)}\right)^\dagger}_{\text{contracts shared indexes}}
              \otimes \
              \underbrace{\MSP(M^{(K-1)})}_{\text{right boundary}} & \mbox{if $K$ is even}.
              \end{array}
    \right.$}
    \label{eq:OPR}
  \end{align}
  The \emph{Two-Tower unitary} is $U_\mathcal{W} = \OP(R)\OP(L)$.
\end{definition}
In $\OP(L)$, each $\SP(M^{(2i)})$ acts on the register pair $(\ket{0}^{l_i},\ket{0}^{r_i})$; all factors target disjoint registers and execute in parallel. In $\OP(R)$, each $\SP(M^{(2j+1)})^\dagger$ acts on the pair $(\ket{\,\cdot\,}^{r_j},\ket{\,\cdot\,}^{l_{j+1}})$ --- the boundary shared between adjacent even-matrix registers --- and, by \Cref{eq:sp-collapse}, contracts it to $\ket{0}$, obtaining the corresponding entry of $M^{(2j+1)}$ as an amplitude up to a normalization factor. The identity operators $\Id_{L_0}$ and $\Id_{R_{(K-1)/2}}$ leave the outer indices $l_0$ and $r_{(K-1)/2}$ untouched; these encode the row and column indices of the product $\mathcal{W}$. The even-$K$ case (see \Cref{fig:circuit-K-even}) requires $\MSP$ rather than $\SP^\dagger$ because $\SP^\dagger$ would contract both registers to $\ket{0}$, leaving no register to carry the output column index of $\mathcal{W}$; $\MSP$ contracts only the shared boundary register while keeping the output register free, yielding the same Frobenius-norm-scaled product. \Cref{alg:skyscraper} summarises the complete procedure; open-source Qiskit and QCLAB implementations are publicly available$^1$. 

\subsubsection{Two-Tower Algorithm Intuition}
\label{sec:algorithm:walkthrough}
The construction draws from a classical result in matrix theory. Lemma~4.3.1 of~\cite{horn2012matrix} expresses the product of three matrices via the Kronecker product and vectorization:
\begin{equation}\label{eq:horn}
AXB = C \quad\Longleftrightarrow\quad (B^T\otimes A)\,\text{vec}(X) = \text{vec}(C).
\end{equation}
In the quantum setting, preparing $A$ and $B$ on separate registers realises their Kronecker product naturally through the tensor-product structure of the Hilbert space. However, \Cref{eq:horn} does not translate directly: $\SP(M)$ encodes any matrix in the statevector as a row-wise vectorization (\Cref{eq:sp}), rather than as a matrix operator as~\Cref{eq:horn} requires for $A$ and $B$. Consequently, whereas the classical formula $(B^T\otimes A)\,\text{vec}(X)$ yields exactly $\text{vec}(C)$, the Two-Tower circuit produces a higher-dimensional statevector containing $\text{vec}(C)$ alongside spurious contributions from index-mismatched combinations of entries of $A$, $X$, and $B$. A further difference is that the Two-Tower does not require any matrix transposition: each matrix is loaded directly through its own state-preparation operator, without reshaping $B$ into $B^T$. The contraction mechanism (\Cref{eq:sp-collapse}) is what separates the two: $\SP(X)^\dagger$ projects the valid, index-matched contributions onto the $\ket{0}$ subspace of the shared registers, while all spurious terms fall into the orthogonal complement $\ket{\bot}$. To illustrate this filtering and observe how the correct sums accumulate, we trace the circuit of \Cref{fig:circuit-3} for $D = ABC$ with $A \in \mathbb{R}^{M \times N}$, $B \in \mathbb{R}^{N \times S}$, $C \in \mathbb{R}^{S \times T}$. Let $m = \log M$, $n = \log N$, $s = \log S$, $t = \log T$, then we define the following quantum registers:
\begin{equation*}\label{eq:psi0-3}
  \ket{\psi_0} = \ket{0}^{m}\ket{0}^{n}\ket{0}^{s}\ket{0}^{t}.
\end{equation*}
\begin{figure}[t]
  \centering
  \begin{quantikz}[wire types={b,b,b,b}, classical gap=0.07cm]
    \lstick{$\ket{0}^{t}$} \slice{$\ket{\psi_0}$}
      & \gate[wires=2]{\SP(C)}\gategroup[4,steps=1,style={dotted, rounded corners, inner sep=1pt},label style={yshift=-5cm}]{$\OP(L)$} \slice{$\ket{\psi_1}$}
      & \gate[style={inner sep=2pt}]{\mathbb{I}_{T}}\gategroup[4,steps=1,style={dotted, rounded corners,inner sep=1pt},label style={yshift=-5cm}]{$\OP(R)$} \slice{$\ket{\psi_2}$}
      & \qw \\
    \lstick{$\ket{0}^{s}$}
      & \qw
      & \gate[wires=2]{\SP(B)^{\dagger}}
      & \qw \\
    \lstick{$\ket{0}^{n}$}
      & \gate[wires=2]{\SP(A)}
      & \qw
      & \qw \\
    \lstick{$\ket{0}^{m}$}
      & \qw
      & \gate[style={inner sep=2pt}]{\mathbb{I}_{M}}
      & \qw
  \end{quantikz}
  \caption{Two-Tower circuit for $D = ABC$ with
    $A \in \mathbb{R}^{M \times N}$, $B \in \mathbb{R}^{N \times S}$,
    $C \in \mathbb{R}^{S \times T}$,
    and $m = \log M$, $n = \log N$, $s = \log S$, $t = \log T$.
    \textbf{$\OP(L)$:} $\SP(C)$ on registers $(t,s)$ and $\SP(A)$
    on registers $(n,m)$ execute in {parallel}.
    \textbf{$\OP(R)$:} $\SP(B)^\dagger$ on registers $(s,n)$ contracts the shared summation index.
    The output state $\ket{\psi_2}$ encodes $D$ in registers
    $(t,m)$ (outer boundaries) with the internal registers $(s,n)$ contracted to $\ket{0}$.}
  \label{fig:circuit-3}
\end{figure}
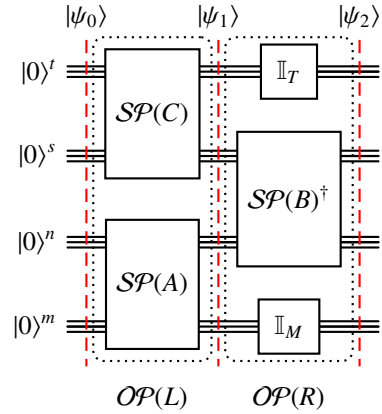
\smallskip\noindent\textbf{Layer 1: apply $\OP(L) = \SP(A) \otimes \SP(C)$.} Operator $\SP(A)$ acts on $(\ket{0}^{m}, \ket{0}^{n})$ while $\SP(C)$ acts on $(\ket{0}^{s}, \ket{0}^{t})$ simultaneously:
\begin{equation*}\label{eq:psi1-3}
\resizebox{\linewidth}{!}{$\displaystyle
  \ket{\psi_1}= \OP(L)\ket{\psi_0}
  = \frac{1}{\norm{A}_F\norm{C}_F}
    \sum_{i=0}^{M-1}\sum_{j=0}^{N-1}\sum_{u=0}^{S-1}\sum_{v=0}^{T-1}
    a_{i,j}\, c_{u,v}\;\ket{i}^{m}\ket{j}^{n}\ket{u}^{s}\ket{v}^{t}.
$}
\end{equation*}
\smallskip\noindent\textbf{Layer 2: apply $\OP(R) = \Id_M \otimes \SP(B)^\dagger \otimes \Id_T$.}
By \Cref{eq:sp-collapse}, $\SP(B)^\dagger$ acts on registers $(n,s)$
picking up $\frac{b_{j, u}}{\norm{B}_F}$ for each basis state $\ket{j,u}$.
Applying to $\ket{\psi_1}$:
\begin{equation*}
\resizebox{\linewidth}{!}{$\displaystyle
\begin{aligned}
  \ket{\psi_2}=\OP(R)\ket{\psi_1}
  &= \frac{1}{\norm{A}_F\norm{B}_F\norm{C}_F}
     \sum_{i=0}^{M-1}\sum_{v=0}^{T-1}
     \Biggl(\sum_{j=0}^{N-1}\sum_{u=0}^{S-1}
       a_{i,j}\, b_{j,u}\, c_{u,v}
     \Biggr)\ket{i,0,0,v}
     + \ket{\bot}.
\end{aligned}
$}
\end{equation*}
The contraction of shared registers via \Cref{eq:sp-collapse} is visible moving from
$\ket{\psi_1}$ to $\ket{\psi_2}$: $\SP(B)^\dagger$ projects each
basis state $\ket{j}^n\ket{u}^s$ onto $\ket{0}^n\ket{0}^s$ with
amplitude $\frac{b_{j,u}}{\norm{B}_F}$.
The indices $j$ and $u$ are summed over in the classical product
$(ABC)_{i, v} = \sum_{j,u} a_{i,j}\,b_{j,u}\,c_{u,v}$.
In the Two-Tower circuit, these sums arise from the superposition
over the shared registers $(n,s)$: the adjoint state preparation
$\SP(B)^\dagger$ converts each basis state $\ket{j,u}$ into the
corresponding matrix entry, and the superposition accumulates all
inner-product terms simultaneously.
As a result, the double sum over $(j,u)$ yields
$d_{i,v} = (ABC)_{i,v}$, confirming that $\ket{\psi_2}$ encodes $D$
in the outer registers $(m,t)$, with the shared registers $(n,s)$
contracted to $\ket{0,0}$.

\section{Analysis}
\label{sec:analysis}
We organize this section as follows: 
\Cref{subsec:correctness} proves that the Two-Tower circuit correctly
encodes the chain product for any $K$.
\Cref{subsec:complexity-analysis} derives the depth and qubit bounds.
\Cref{subsec:normalisation-signal} discusses the normalisation factor and 
the signal weight.

\subsection{Correctness}\label{subsec:correctness}
The correctness proof relies on the contraction mechanism of
\Cref{eq:sp-collapse}, established in \Cref{sec:chain_mult}.
The main result follows.

\begin{theorem}[Two-Tower Matrix Multiplication]\label{thm:main}
  Let $\mathcal{W} = M^{(0)}, \ldots, M^{(K-1)}$ be a matrix chain with $M^{(k)}\in \mathbb{R}^{P_k \times P_{k+1}}$ for $k = 0, \ldots, K-1$,
  and with all $p_k=\log_2P_k$.
  There exists a unitary $U_\mathcal{W} = \OP(R)\,\OP(L)$ such that
  \begin{equation}\label{eq:main-mapping}
  \begin{aligned}
    &U_\mathcal{W} \;\colon\;
    \ket{0}^{p_0}\ket{0}^{p_1+\cdots+p_{K-1}}\ket{0}^{p_K}
    \\
    &\;\longmapsto\;
    \frac{1}{\prod_{k=0}^{K-1}\norm{M^{(k)}}_F}
    \sum_{p=0}^{P_0-1}\sum_{q=0}^{P_K-1}
      w_{p,q}\,
      \ket{p}^{p_0}\ket{0}^{p_1+\cdots+p_{K-1}}\ket{q}^{p_{K}}
    + \ket{\bot},
  \end{aligned}
  \end{equation}
  where $\mathcal{W} = M^{(0)} \cdots M^{(K-1)}$ and $w_{p,q} = (\mathcal{W})_{p,q}$.
  The term $\ket{\bot}$ is a superposition
  of basis states in which at least one intermediate (shared) register
  is not $\ket{0}$; it is orthogonal to the subspace
  $\mathcal{H}_{\mathcal{W}} = \mathrm{span}\{\lvert p,0,\ldots,0,q \rangle \}$ for $0\le p < P_0,\; 0\le q < P_K$
  The circuit $U_\mathcal{W}$ has:
  \begin{itemize}
    \item \textbf{Depth}
          $\mathcal{O}\!\bigl(\max_{k=0}^{K-1} {\,\mathrm{polylog}}(P_k P_{k+1})\bigr)$, independent of $K$;
    \item \textbf{Qubits}
          $O\!\bigl(\sum_{k=0}^{K} \log P_k\bigr)$.
  \end{itemize}
\end{theorem}

\begin{proof}[Proof of \Cref{thm:main} for odd $K$]\label{subsec:proof-odd}
Using the notation of \Cref{sec:chain_mult}: left matrices
$M^{(2i)} \in \mathbb{R}^{P_k \times P_{k+1}}$ (where $k = 2i$) and right matrices
$M^{(2j+1)} \in \mathbb{R}^{ P_{k+1} \times  P_{k+2}}$  (where $k = 2j$), with shared registers shared between adjacent pairs. 

\smallskip\noindent\textbf{Step~1.}
Since all factors of $\OP(L) = \bigotimes_i \SP(M^{(2i)})$ act on
disjoint registers:
\begin{equation*}\label{eq:psi1-gen}
\resizebox{\linewidth}{!}{$\displaystyle
  \ket{\psi_1} = \frac{1}{\prod_i \norm{M^{(2i)}}_F}
    \!\!\sum_{\vec{l},\vec{r}}
    \!\!\left(\prod_i m^{(2i)}_{l_i,r_i}\right)
    \ket{l_0}\ket{r_0,\,l_1,\, r_1,\ldots,r_{(K-1)/2-1}\,l_{(K-1)/2}}\ket{r_{(K-1)/2}}.
$}
\end{equation*}

Here $\vec{l} = (l_0, \ldots, l_{(K-1)/2})$ and
$\vec{r} = (r_0, \ldots, r_{(K-1)/2})$ collect all row and column indices of
the left matrices $M^{(2i)}$ (so $L_i = P_{2i}$ and $R_i = P_{2i+1}$), with
$0 \leq l_i < L_i$ and $0 \leq r_i < R_i$. Of these $K+1$ indices, the two
outmost ones, $l_0$ and $r_{(K-1)/2}$, carry the output row and column of
$\mathcal{W}$; the remaining shared indices
$r_0, l_1, r_1, \ldots, r_{(K-1)/2-1}, l_{(K-1)/2}$ are the ones
contracted by $\OP(R)$.

\smallskip\noindent\textbf{Step~2.}
By \Cref{eq:sp-collapse}, each $\SP(M^{(2j+1)})^\dagger$ on shared pair
$\ket{r_j, l_{j+1}}$ contributes amplitude
$\frac{m^{(2j+1)}_{r_j,l_{j+1}}}{\norm{M^{(2j+1)}}_F}$ on $\ket{0,0}$.
Applying all $(K-1)/2$ such operators in parallel and collecting:
\begin{equation*}\label{eq:OPR-psi1-gen}
\resizebox{\linewidth}{!}{$\displaystyle
  \OP(R)\ket{\psi_1}
  = \frac{1}{\prod_{k} \norm{M^{(k)}}_F}
    \sum_{\vec{l},\vec{r}}
    \!\!\left(\prod_i m^{(2i)}_{l_i,r_i}\right)
    \!\!\left(\prod_j m^{(2j+1)}_{r_j,l_{j+1}}\right)
    \ket{l_0,\underbrace{0,\ldots,0}_{\text{shared}},r_{(K-1)/2}}
    + \ket{\bot}.
$}
\end{equation*}
\noindent Observing that
\[
\resizebox{\linewidth}{!}{$\displaystyle
   \sum_{\vec{l},\vec{r}}
    \!\!\left(\prod_{i} m^{(2i)}_{l_i,r_i}\right)
    \!\!\left(\prod_j m^{(2j+1)}_{r_j,l_{j+1}}\right)=  \!\!\sum_{r_0,l_1,\ldots,l_{(K-1)/2}}
  \!\!m^{(0)}_{l_0,r_0}\cdots m^{(K-1)}_{l_{(K-1)/2},r_{(K-1)/2}}
  = w_{l_0,\,r_{(K-1)/2}}\;,
$}
\]
and setting
$p \leftarrow l_0$, $q \leftarrow r_{(K-1)/2}$
yields~\Cref{eq:main-mapping}.
\end{proof}


When $K$ is even, the chain contains $K/2$ even-indexed matrices ${M^{(0)}, M^{(2)}, \ldots, M^{(K-2)}}$ and $K/2$ odd-indexed matrices ${M^{(1)}, M^{(3)}, \ldots, M^{(K-1)}}$. The last odd-indexed matrix $M^{(K-1)}$ has no paired right neighbour, so a plain $\SP(M^{(K-1)})^{\dagger}$ would contract both registers to $\ket{0}$, leaving no register to carry the output column index. Instead, we encode $M^{(K-1)}$ via the modular state-preparation operator $\MSP(M^{(K-1)})$.

\begin{proof}[Proof of \Cref{thm:main} for even $K$]\label{subsec:proof-even}
Apply the odd-$K$ argument of \Cref{subsec:proof-odd} to the first $K-1$
matrices (an odd-length chain).
Then, we append a fresh ancilla register of $p_K=\log_2 P_K$ qubits for the
final column index.
The intermediate state after processing the first $K-1$ matrices of the chain $Z= M^{(0)}\cdots M^{(K-2)}$ is:
\begin{equation*}\label{eq:psi-half}
   \frac{1}{\prod_{k=0}^{K-2}\norm{M^{(k)}}_F}
    \sum_{l_0=0}^{P_0-1}\sum_{d=0}^{P_{K-1}-1}
    z_{l_0,d}\,
\ket{l_0,\underbrace{0,\ldots,0}_{\text{shared}},d,\underbrace{0}_{\text{ancilla}}}
    \;+\; \ket{\bot},
\end{equation*}

where $z_{l_0,d} = (M^{(0)}\cdots M^{(K-2)})_{l_0,d}$ and the ancilla register remains in state $\ket{0}$. On the outer registers $(p_0,p_{K-1})$ the amplitudes
are exactly the entries of $Z$, so the state coincides with the output of
$\SP(Z)$ {up to} the subnormalisation factor $\lVert Z\rVert_F\big/\prod_{k=0}^{K-2}\lVert M^{(k)}\rVert_F\le 1$ (the tower carries the product normalisation $\prod_{k=0}^{K-2}\lVert M^{(k)}\rVert_F$ rather than $\lVert Z\rVert_F$; see \Cref{subsec:normalisation-signal}). This scalar does not affect the final step: applying $\MSP(M^{(K-1)})$ as in the two-matrix
circuit of \Cref{sec:two_matrix} (since $\mathcal{W}=ZM^{(K-1)}$) contracts the
shared index and multiplies in the entries of $M^{(K-1)}$, and the normalization factors combine to $\prod_{k=0}^{K-1}\lVert M^{(k)}\rVert_F$, yielding \Cref{eq:main-mapping}.
\end{proof}


\subsection{Complexity}\label{subsec:complexity-analysis}
The depth and qubit bounds stated in \Cref{thm:main} derive from the structure of the two-layer circuit and the state-preparation cost. In particular:

\begin{itemize}
  \item \textbf{Depth.} Each $\SP(M^{(k)})$ and $\SP(M^{(k)})^\dagger$ has depth $\mathcal{O}({\mathrm {polylog}}(P_kP_{k+1}))$. Since all even-indexed $\SP$ operators execute in parallel on the first layer, and all odd-indexed $\SP^\dagger$ operators execute in parallel on the second layer, the total depth is $\mathcal{O}(\max_{k} {\mathrm {polylog}}(P_k P_{k+1}))$, independent of $K$.
  \item \textbf{Qubits.} The circuit allocates one register of $p_k=\log_2 P_k$ qubits for each matrix $M^{(k)}$, $k=0,\ldots,K-1$, and a final register  $p_K=\log_2 P_K$ qubits for the output column index. Adjacent matrices share a shared register already counted in the allocation for $M^{(k+1)}$, and no additional ancillae are required. The total is $\sum_{k=0}^{K} \log_2 P_k = \Theta\!\bigl(\sum_{k=0}^{K} \log_2 P_k \bigr)$. 
  \item \textbf{Total gate count.} The two layers apply $K$ state-preparation operators overall, so the total gate count is $\Theta\bigl(\sum_{k=0}^{K-1} \mathrm{polylog}(P_k P_{k+1})\bigr)$ which becomes  $\Theta(K\,\mathrm{polylog}(N))$ for square $N \times N$ matrices.
\end{itemize}

\subsection{Normalisation and signal weight}\label{subsec:normalisation-signal}
The amplitudes in~\Cref{eq:main-mapping} carry the {normalisation factor} $\mathcal{N} = \frac{1}{\prod_{k=0}^{K-1} \norm{M^{(k)}}_F}$, which we compute classically before running the circuit. The residual term $\ket{\bot}$ lies in the orthogonal complement of $\mathcal{H}_{\mathcal{W}} = \mathrm{span}\{ \lvert p,0,\ldots,0,q \rangle \}$ for $0 \leq p < P_0$ and $0 \leq q < P_K$ and has squared norm $1 - \norm{\mathcal{W}}_F^2 \mathcal{N}^2$; this term does not indicate a failure of the subroutine but is the component of the unitary evolution that falls outside the subspace encoding the product. We therefore define the {signal weight} as the total squared amplitude carried by the basis states in $\mathcal{H}_{\mathcal{W}}$ --- those whose amplitudes, up to the normalisation factor $\mathcal{N}$, are exactly the entries of the chain product:
$$
\Omega_{\mathcal{W}}
= \frac{\norm{\mathcal{W}}_F^2}{\prod_{k=0}^{K-1} \norm{M^{(k)}}_F^2}.
$$
In a standalone algorithm, $\Omega_{\mathcal{W}}$ would correspond to the success probability of post-selecting on $\mathcal{H}_\mathcal{W}$. However, the Two-Tower --- much like the QFT --- is designed as a subroutine with a downstream algorithm applied before any measurement; we therefore use the term {signal weight} to emphasise that $\Omega_{\mathcal{W}}$ is a parameter inherited by the subsequent computation, not a measurement outcome of the circuit itself. 

By sub-multiplicativity of the Frobenius norm
($\norm{AB}_F \leq \norm{A}_F \norm{B}_F$),
$\Omega_{\mathcal{W}} \leq 1$, with $\Omega_{\mathcal{W}} = 1$ if and only
if $\|\mathcal{W}\|_F= \prod_{k=0}^{K-1} \|M^{(k)}\|_F$, and this happens in special cases. For example, consider $K$ copies of a normal $N \times N$ matrix $A$, i.e., a matrix such that  $A^\dagger A=AA^\dagger$ and denote its singular values by
$\sigma_1 \geq \cdots \geq \sigma_N \ge 0$.
Since $\norm{A}_F^2 = \sum_i \sigma_i^2$ and singular values of $\mathcal{W} = A^K$ when $A$ is normal are $\sigma_i^K$,  we get  
$\norm{A^K}_F^2 = \sum_i \sigma_i^{2K}$.  The signal weight becomes
$$
  \Omega_{\mathcal{W}}
  = \frac{\norm{A^K}_F^2}{\norm{A}_F^{2K}}
  = \frac{\sum_i \sigma_i^{2K}}{(\sum_i \sigma_i^2)^K},
$$
which for large $K$ is dominated by the largest singular value:
$\Omega_{\mathcal{W}} \approx (\sigma_1^2 / \norm{A}_F^2)^{K-1}$.
In general, $\Omega_{\mathcal{W}}$ can decrease with $K$; the rate of decay depends on the spectral structure of the matrices.
Two edge cases are instructive.
\begin{itemize}
  \item \emph{Well-conditioned matrices}
    ($\sigma_1/\sigma_N \approx 1$): all singular values are
    approximately equal to the same  $\sigma$, so
    $\|A^K\|_F^2 / \norm{A}_F^2 \approx N\,\sigma^{2K} / (N\sigma^2)^K
    $, giving $\Omega_{\mathcal{W}} \approx N^{-(K-1)}$. The signal weight decays exponentially with $K$: the information spreads across many directions and repeated multiplication reduces the overall signal strength. Note that the unitary case ($\sigma = 1$) belongs to this category, since $\Omega_{\mathcal{W}} = N/N^K = 1/N^{K-1}$.
  \item \emph{Near-rank-one matrices}
    ($\sigma_2 \approx 0$): a single singular value dominates the Frobenius norm, so $\Omega_{\mathcal{W}} \approx \frac{\sigma_1^{2K}}{\norm{A}_F^{2K}} \approx 1$. The matrix acts essentially as a scalar along one direction, and repeated multiplications do not weaken the signal.
\end{itemize}
The signal weight is therefore data-dependent: matrices with flat spectra cause rapid decay, whereas matrices with peaked spectra maintain $\Omega_{\mathcal{W}}$ close to unity. This phenomenon is counterintuitive: in most numerical settings, well-conditioned matrices are the favorable case, but here a flat spectrum spreads the signal uniformly across all singular directions, diluting it exponentially with $K$. Nevertheless, we observe that Amplitude Amplification~\cite{brassard2002quantum} can boost $\Omega_{\mathcal{W}}$ to any desired constant, at the cost of $\mathcal{O}(1/\sqrt{\Omega_{\mathcal{W}}})$ calls to the Two-Tower circuit and its inverse, introducing a dependence on $K$. 
\begin{figure}[t]    
  \centering
  \begin{quantikz}[wire types={b,b,b,b}, classical gap=0.07cm]
    \lstick{$\ket{0}^{t}$}
      & \gate[wires=2]{\SP(C)}\gategroup[4,steps=1,style={dotted, rounded corners, inner sep=1pt},label style={yshift=-5cm}]{$\OP(L)$}
      & \gate[style={inner sep=2pt}]{\SP(v)^{\dagger}}\gategroup[4,steps=1,style={dotted, rounded corners, inner sep=1pt},label style={yshift=-5cm}]{$\OP(R)$}
      & \qw \\
    \lstick{$\ket{0}^{s}$}
      & \qw
      & \gate[wires=2]{\SP(B)^{\dagger}}
      & \qw \\
    \lstick{$\ket{0}^{n}$}
      & \gate[wires=2]{\SP(A)}
      & \qw
      & \qw \\
    \lstick{$\ket{0}^{m}$}
      & \qw
      & \gate[style={inner sep=2pt}]{\mathbb{I}_{M}}
      & \qw
  \end{quantikz}    
  \caption{Chain-vector circuit for $w = ABCv$ with
    $A \in \mathbb{R}^{M \times N}$, $B \in \mathbb{R}^{N \times S}$,
    $C \in \mathbb{R}^{S \times T}$, $v \in \mathbb{R}^{T \times 1}$.
    Since $P_K= 1$, the output register vanishes and the
    $\MSP$ decomposition of the last factor reduces to
    $\SP(v)^\dagger$ on the $t$-qubit boundary register alone.
    The product $w = ABCv \in \mathbb{R}^{M}$ is encoded
    in the $m$-qubit register, with all shared registers $(t,s,n)$
    contracted to $\ket{0}$.}
    \label{fig:circuit-vector}
\end{figure}
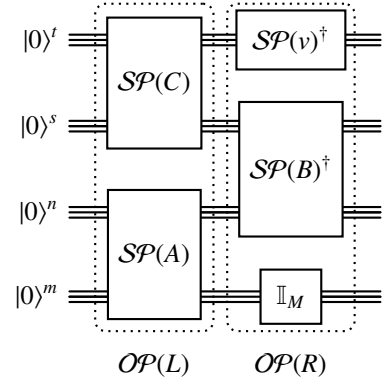

\section{Matrix-Chain-Vector Multiplication}
\label{sec:vector}
An important special case of \Cref{thm:main} arises when the last matrix in the chain is a column vector. This case finds direct application in iterative solvers, the power method, and Krylov subspace methods~\cite{GolubVanLoan1996}, all of which require repeated matrix-vector products. The following corollary formalises this specialisation.

\begin{corollary}[Matrix-chain-vector multiplication]\label{cor:vector-case}
  Setting $P_K = 1$ in \Cref{thm:main} identifies $M^{(K-1)}$ with
  a column vector $v \in \mathbb{R}^{P_{K-1}}$.
  The right-boundary register vanishes ($\log 1 = 0$ qubits), and
  the Two-Tower encodes the product
  $w = M^{(0)}\cdots M^{(K-2)}\,v \in \mathbb{R}^{P_0}$
  as the amplitudes of the $\log P_0$-qubit output register, in
  depth $\mathcal{O}(\max_k {\mathrm {polylog}}(P_k P_{k+1}))$ using
  $\Theta\!\bigl(\sum_{k=0}^{K-1} \log P_k \bigr)$
  qubits.
\end{corollary}
\begin{proof}
  Immediate from \Cref{thm:main} with $P_K = 1$: the
  $\log_2 P_K= 0$ qubits vanish from both the register allocation
  and the output state, and all remaining bounds carry over unchanged.
\end{proof}

We validate this specialisation at machine precision using open-source Qiskit and QCLAB implementations$^1$, confirming that the circuit correctly produces the chain-vector product for all tested instances. \Cref{fig:circuit-vector} illustrates the case $K=4$ with $w = ABCv$: the column-index register is no longer needed and $\MSP(v)$ reduces to $\SP(v)^\dagger$ on the boundary register alone, so the product vector is encoded entirely in the bottom register.

\section{Conclusions and Future Work}
\label{sec:conclusions}
We presented the Two-Tower Matrix Multiplication, a quantum subroutine that encodes $\mathcal{W} = M^{(0)}\cdots M^{(K-1)}$, $M^{(k)}\in \mathbb{R}^{P_k\times P_{k+1}}$, using $\Theta\bigl(\sum_{k=0}^{K} \log P_k \bigr)$ qubits and circuit depth $\mathcal{O}(\max_k \mathrm{polylog}(P_k P_{k+1}))$, independent of $K$. The construction generalises the two-matrix operators of~\cite{bernasconiMatMul} to arbitrary chain length by interleaving $\SP$ and $\SP^\dagger$ gates on disjoint registers; the $\SP^\dagger$ contraction drives all shared registers to $\ket{0}$, producing the chain product as amplitudes.
Compared with classical algorithms, the Two-Tower trades sequential $\mathcal{O}(K \cdot N^\omega)$ depth --- or $\mathcal{O}(\log K \cdot \log N)$ with unbounded parallelism --- for $\mathcal{O}(\mathrm{polylog}(N))$ quantum depth. The advantage materialises when the downstream computation requires only aggregate quantities of $\mathcal{W}$ extractable via Amplitude Estimation. Representative instances include: norm, trace, and variance estimation via Amplitude Estimation~\cite{antonioli2025outlier,bernasconi2024quantum,cade2017quantum,poggiali2023quantum};
graph-matrix powers $A^K$ for spectral analysis, walk counting, and
mixing estimation~\cite{nghiem2023quantum,janzing610235bqp};
linear systems with chain-structured operators via
HHL~\cite{harrow2009quantum} or Neumann-series
LCU~\cite{chakraborty_et_al:LIPIcs.ICALP.2019.33};
and kernel evaluation in quantum machine
learning~\cite{biamonte2017quantum,kerenidis_et_al:LIPIcs.ITCS.2017.49}.
The chain-vector specialisation opens connections to iterative
numerical methods (including the power method ($A^t v$), Krylov
subspace construction $\{v, Av, \ldots, A^{t-1}v\}$, iterative
refinement via Neumann series, and polynomial matrix evaluation
$p(A)v$) where the Two-Tower provides a depth-efficient
building block for each matrix-vector product in the iteration.

Several directions for future work emerge from the present
construction.

\begin{itemize}
\item \emph{Trading depth for width.}
As shown in \Cref{subsec:normalisation-signal}, the signal weight $\Omega_{\mathcal{W}}$ can decrease with the chain length $K$. A natural mitigation is a {hybrid strategy} that alternates between the Two-Tower and block-encoding: the Two-Tower computes partial chain products in constant depth at the cost of growing qubit count and decaying signal weight, while block-encoding stages absorb the intermediate results and continue the multiplication in depth rather than width. This approach trades qubit overhead for circuit depth within the same computation, but it requires an efficient conversion between statevector encoding and block encoding. Such a conversion is particularly appealing beyond its role in the hybrid strategy itself: by lifting an intermediate statevector into an operator, i.e. into a block encoding, one could subsequently act on it with the Quantum Singular Value Transformation (QSVT), thereby inheriting the family of algorithms and applications already developed for block-encoded operators. We leave the design of such a conversion procedure, and the resulting extension of the hybrid strategy, as an open problem for future investigation.

\item \emph{Complex matrices.}
The current construction assumes real entries. Extending to complex-valued matrices requires loading the entry-wise conjugate $\overline{M^{(k)}}$ in place of $M^{(k)}$ within $\OP(R)$, so that the contraction mechanism produces the correct Hermitian inner products; the rest of the circuit remains unchanged. A formal proof and numerical validation are left for future work.

\item \emph{Unitary chain encoding.}
An interesting generalisation replaces the matrices $M^{(0)}, M^{(1)},\ldots,$ $ M^{(K-1)}$ with arbitrary unitary operators $U_0, \ldots, U_{K-1}$, encoding their product as a quantum state. Each $U_k$ would enter the circuit through a state-preparation oracle that places the vectorized unitary in its first column. If the contraction mechanism extends to this setting, the Two-Tower could compute transition amplitudes $\bra{\phi} U_0 \cdots U_{K-1} \ket{v}$ at depth independent of $K$, with direct applications in quantum simulation and process tomography. We leave the formal analysis of this direction for future work.
\end{itemize}

\section*{Acknowledgements}
This study was carried out within the National Centre on HPC, Big Data and Quantum Computing - SPOKE 10 (Quantum Computing) and received funding from the European Union Next-GenerationEU - National Recovery and Resilience Plan (NRRP) – MISSION 4 COMPONENT 2, INVESTMENT N. 1.4 – CUP N. I53C22000690001. 
Partial support was also provided by the INdAM - GNCS project ``Algebra Lineare Quantistica, State Preparation e Compilazione di Circuiti Quantistici'', CUP E53C25002010001, and by the Italian Project Fondo Italiano per la Scienza FIS00001966 ``MIMOSA''.
G. Del Corso is also partially supported by European Union - NextGenerationEU under the National Recovery and Resilience Plan (PNRR) - Mission 4 Education and research - Component 2 From research to business - Investment 1.1 Notice Prin 2022 - DD N. 104 2/2/2022, titled Low-rank Structures and Numerical Methods in Matrix and Tensor Computations and their Application, proposal code 20227PCCKZ – CUP I53D23002280006.J53D23003620006 and by the U.S. Department of Energy, Office of Science, National Quantum Information Science Research Centers, Superconducting Quantum Materials and Systems Center (SQMS) under Contract No. 89243024CSC000002.

\section*{Declaration of generative AI and AI-assisted technologies in the manuscript preparation process.}
During the preparation of this work, the authors used Claude Opus 4.6  for  grammar, spelling check and help generating part of the circuits drawing. After using these tool, the authors reviewed and edited the content as needed and take full responsibility for the publication’s content. 
\bibliographystyle{elsarticle-num}
\bibliography{main}

@misc{qiskit2024,
      title={Quantum computing with {Q}iskit},
      author={A. Javadi-Abhari and M. Treinish and K. Krsulich and C. J. Wood and J. Lishman and J. Gacon and S. Martiel and  P. D. Nation and L. S. Bishop and A. W. Cross and B. R. Johnson and J. M. Gambetta},
      year={2024},
      doi={10.48550/arXiv.2405.08810},
      eprint={2405.08810},
      archivePrefix={arXiv},
      primaryClass={quant-ph}
}

@article{brassard2002quantum,
  title     = {Quantum amplitude amplification and estimation},
  author    = {Brassard, G. and H{\o}yer, P. and Mosca, M. and Tapp, A.},
  journal   = {Contemporary Mathematics},
  volume    = {305},
  pages     = {53--74},
  year      = {2002},
  publisher = {American Mathematical Society},
  note      = {arXiv:quant-ph/0005055}
}

@inproceedings{poggiali2023quantum,
  title     = {Quantum Feature Selection with Variance Estimation},
  author    = {Poggiali, A. and Bernasconi, A. and Berti, A. and {Del Corso}, G. M. and Guidotti, R. and others},
  booktitle = {ESANN},
  year      = {2023}
}

@inproceedings{montanaro2024quantum,
  title     = {Quantum and classical query complexities of functions of matrices},
  author    = {Montanaro, A. and Shao, C.},
  booktitle = {Proceedings of the 56th Annual ACM Symposium on Theory of Computing},
  pages     = {573--584},
  year      = {2024}
}

@article{li2025faster,
  title     = {Faster quantum subroutine for matrix chain multiplication via {Chebyshev} approximation},
  author    = {Li, X. and Zheng, P.-L. and Pan, C. and Wang, F. and Cui, C. and Lu, X.},
  journal   = {Scientific Reports},
  volume    = {15},
  number    = {1},
  pages     = {28559},
  year      = {2025},
  publisher = {Nature Publishing Group UK London}
}

@article{nghiem2023quantum,
  title     = {Quantum algorithm for estimating largest eigenvalues},
  author    = {Nghiem, N. A. and Wei, T.-C.},
  journal   = {Physics Letters A},
  volume    = {488},
  pages     = {129138},
  year      = {2023},
  publisher = {Elsevier}
}

@article{janzing610235bqp,
  title   = {{BQP}-complete problems concerning mixing properties of classical random walks on sparse graphs},
  author  = {Janzing, D. and Wocjan, P.},
  journal = {arXiv preprint quant-ph/0610235},
  year    = {2006}
}

@article{cade2017quantum,
  title   = {The quantum complexity of computing {Schatten} $p$-norms},
  author  = {Cade, C. and Montanaro, A.},
  journal = {arXiv preprint arXiv:1706.09279},
  year    = {2017}
}

@article{berti2025efficient,
  title   = {Efficient quantum state preparation with bucket brigade {QRAM}},
  author  = {Berti, A. and Ghisoni, F.},
  journal = {arXiv preprint arXiv:2510.16149},
  year    = {2025}
}

@inproceedings{antonioli2025outlier,
  title        = {Outlier Detection and other applications of Quantum Matrix Multiplication},
  author       = {Antonioli, G. and Berti, A. and Poggiali, A. and Bernasconi, A. and {Del Corso}, G. M.},
  booktitle    = {2025 IEEE International Parallel and Distributed Processing Symposium Workshops (IPDPSW)},
  pages        = {509--518},
  year         = {2025},
  organization = {IEEE}
}

@book{nielsen2010quantum,
  title     = {Quantum Computation and Quantum Information},
  author    = {Nielsen, M. A. and Chuang, I. L.},
  year      = {2010},
  publisher = {Cambridge University Press}
}

@article{bernasconiMatMul,
  title   = {Quantum Subroutine for Efficient Matrix Multiplication},
  author  = {Bernasconi, A. and Berti, A. and {Del Corso}, G. M. and Poggiali, A.},
  journal = {IEEE Access},
  volume  = {12},
  pages   = {116274--116284},
  year    = {2024},
  doi     = {10.1109/ACCESS.2024.3446176}
}

@article{bernasconi2024quantum,
  title     = {Quantum subroutine for variance estimation: algorithmic design and applications},
  author    = {Bernasconi, A. and Berti, A. and {Del Corso}, G. M. and Guidotti, R. and Poggiali, A.},
  journal   = {Quantum Machine Intelligence},
  volume    = {6},
  number    = {2},
  pages     = {78},
  year      = {2024},
  publisher = {Springer}
}

@article{biamonte2017quantum,
  title     = {Quantum machine learning},
  author    = {Biamonte, J. and Wittek, P. and Pancotti, N. and Rebentrost, P. and Wiebe, N. and Lloyd, S.},
  journal   = {Nature},
  volume    = {549},
  number    = {7671},
  pages     = {195--202},
  year      = {2017},
  publisher = {Nature Publishing Group}
}

@inproceedings{chakraborty_et_al:LIPIcs.ICALP.2019.33,
  title     = {{The Power of Block-Encoded Matrix Powers: Improved Regression Techniques via Faster Hamiltonian Simulation}},
  author    = {Chakraborty, S. and Gily\'{e}n, A. and Jeffery, S.},
  booktitle = {46th International Colloquium on Automata, Languages, and Programming (ICALP 2019)},
  series    = {Leibniz International Proceedings in Informatics (LIPIcs)},
  volume    = {132},
  pages     = {33:1--33:14},
  year      = {2019},
  publisher = {Schloss Dagstuhl -- Leibniz-Zentrum f{\"u}r Informatik},
  doi       = {10.4230/LIPIcs.ICALP.2019.33}
}

@inproceedings{kerenidis_et_al:LIPIcs.ITCS.2017.49,
  title     = {{Quantum Recommendation Systems}},
  author    = {Kerenidis, I. and Prakash, A.},
  booktitle = {8th Innovations in Theoretical Computer Science Conference (ITCS 2017)},
  series    = {Leibniz International Proceedings in Informatics (LIPIcs)},
  volume    = {67},
  pages     = {49:1--49:21},
  year      = {2017},
  publisher = {Schloss Dagstuhl -- Leibniz-Zentrum f{\"u}r Informatik},
  doi       = {10.4230/LIPIcs.ITCS.2017.49}
}

@article{PhysRevLett.120.050502,
  title     = {Quantum Linear System Algorithm for Dense Matrices},
  author    = {Wossnig, L. and Zhao, Z. and Prakash, A.},
  journal   = {Physical Review Letters},
  volume    = {120},
  number    = {5},
  pages     = {050502},
  year      = {2018},
  publisher = {American Physical Society},
  doi       = {10.1103/PhysRevLett.120.050502}
}

@article{harrow2009quantum,
  title     = {Quantum algorithm for linear systems of equations},
  author    = {Harrow, A. W. and Hassidim, A. and Lloyd, S.},
  journal   = {Physical Review Letters},
  volume    = {103},
  number    = {15},
  pages     = {150502},
  year      = {2009},
  publisher = {APS}
}

@inproceedings{alman2021refined,
  title        = {A refined laser method and faster matrix multiplication},
  author       = {Alman, J. and Williams, V. V.},
  booktitle   = {Proceedings of the 2021 ACM-SIAM Symposium on Discrete Algorithms (SODA)},
  pages        = {522--539},
  year         = {2021},
  organization = {SIAM}
}

@article{giovannetti2008quantum,
  title     = {Quantum random access memory},
  author    = {Giovannetti, V. and Lloyd, S. and Maccone, L.},
  journal   = {Physical Review Letters},
  volume    = {100},
  number    = {16},
  pages     = {160501},
  year      = {2008},
  publisher = {APS}
}

@inproceedings{duan2023faster,
  title        = {Faster matrix multiplication via asymmetric hashing},
  author       = {Duan, R. and Wu, H. and Zhou, R.},
  booktitle    = {Proceedings of the 64th Annual IEEE Symposium on Foundations of Computer Science (FOCS)},
  pages        = {2129--2138},
  year         = {2023},
  organization = {IEEE}
}

@book{horn2012matrix,
  title={Matrix analysis},
  author={Horn, R. A and Johnson, C. R.},
  year={2012},
  publisher={Cambridge university press}
}

@article{Fang_2023,
   title={Time-marching based quantum solvers for time-dependent linear differential equations},
   volume={7},
   ISSN={2521-327X},
   url={http://dx.doi.org/10.22331/q-2023-03-20-955},
   DOI={10.22331/q-2023-03-20-955},
   journal={Quantum},
   publisher={Verein zur Forderung des Open Access Publizierens in den Quantenwissenschaften},
   author={Fang, D. and Lin, L. and Tong, Y.},
   year={2023},
   month=mar, pages={955} }

@article{Bini1979,
  title={$O(n^2.7799)$ Complexity for $n\times n$ Approximate Matrix Multiplication},
  author={ D. Bini and M.  Capovani and F. Romani  and G.
  Lotti },
  journal={Information Processing Letters},
  year={1979},
  volume={8},
  pages={234-235},
  doi={10.1016/0020-0190(79)90113-3}
}

@inproceedings{keip2025qclab,
  title={{QCLAB}: A Matlab Toolbox for Quantum Computing},
  author={S. Keip and Camps, D. and Van Beeumen, R.},
  booktitle={2025 IEEE International Parallel and Distributed Processing Symposium Workshops (IPDPSW)},
  pages={1175--1181},
  year={2025},
  organization={IEEE}
}

@inproceedings{gilyen2019quantum,
  author    = {Gily{\'{e}}n, A. and
               Su, Y. and
               Low, G. H. and
               Wiebe, N.},
  title     = {Quantum Singular Value Transformation and Beyond:
               Exponential Improvements for Quantum Matrix Arithmetics},
  booktitle = {Proceedings of the 51st Annual {ACM} {SIGACT} Symposium
               on Theory of Computing ({STOC} 2019)},
  pages     = {193--204},
  publisher = {{ACM}},
  year      = {2019},
  doi       = {10.1145/3313276.3316366},
 }

@Book{GolubVanLoan1996,
  author = "G. H. Golub and C. F. {Van~Loan}",
  title = "Matrix computations",
  publisher = "Johns Hopkins University Press",
  year = "1996",
  edition = "Third",
  address = "Baltimore, MD",
  ISBN = "0-8018-5414-8"
}

@article{poggiali2026more,
  title={A more efficient quantum circuit for estimating the variance},
  author={Poggiali, A. and Ju, J.},
  journal={Quantum Machine Intelligence},
  volume={8},
  number={1},
  pages={34},
  year={2026},
  publisher={Springer}
}

\end{document}